\documentclass[11pt]{article} 
\pdfoutput=1

\usepackage{hyperref}
\usepackage{fullpage}

\usepackage{amsthm}
\usepackage{eucal}
\usepackage{dsfont}
\usepackage{libertine}
\usepackage{newtxmath}
\usepackage[capitalise]{cleveref}
\usepackage{thm-restate}
\usepackage{subfig}
\usepackage{environ}
\usepackage{tikzsymbols}
\usepackage[normalem]{ulem}
\usepackage{multirow}
\usepackage{afterpage}

\makeatletter
\renewcommand{\section}{\@startsection {section}{1}{\z@}%
             {-3.5ex \@plus -1ex \@minus -.2ex}%
             {2.3ex \@plus .2ex}%
             {\normalfont\Large\scshape\bfseries}}
\renewcommand{\subsection}{\@startsection{subsection}{2}{\z@}%
             {-3.25ex\@plus -1ex \@minus -.2ex}%
             {1.5ex \@plus .2ex}%
             {\normalfont\large\scshape\bfseries}}
\renewcommand{\subsubsection}{\@startsection{subsubsection}{2}{\z@}%
             {-3.25ex\@plus -1ex \@minus -.2ex}%
             {1.5ex \@plus .2ex}%
             {\normalfont\normalsize\scshape\bfseries}}
\makeatother

\allowdisplaybreaks[1]

\newcommand{\Description}[1] {}

\usepackage{mathtools}

\usepackage{tikz}
\usetikzlibrary{arrows}
\usepackage{qcircuit}

\theoremstyle{plain}
\newtheorem{theorem}{Theorem}[section]
\newtheorem{lemma}[theorem]{Lemma}

\newtheorem{cor}[theorem]{Corollary}

\theoremstyle{definition}
\newtheorem{definition}[theorem]{Definition}

\newtheorem{remark}[theorem]{Remark}
\newtheorem{fact}[theorem]{Fact}
\newtheorem{example}[theorem]{Example}

\newcommand {\minusspace} {\: \! \!}

\newcommand {\Fn} [2] {\ensuremath{ #1 \minusspace \Br{ #2 } }}

\newcommand {\set} [1] {\ensuremath{ \left\lbrace #1 \right\rbrace }}

\newcommand{\normthree}[1]{{\left\vert\kern-0.25ex\left\vert\kern-0.25ex\left\vert #1 \right\vert\kern-0.25ex\right\vert\kern-0.25ex\right\vert}}
\newcommand {\br} [1] {\ensuremath{ \left( #1 \right) }}
\newcommand {\Br} [1] {\ensuremath{ \left[ #1 \right] }}

\newcommand {\norm} [1] {\ensuremath{ \left\| #1 \right\| }}
\newcommand {\normsub} [2] {\ensuremath{ \norm{#1}_{#2} }}

\newcommand {\abs} [1] {\ensuremath{ \left| #1 \right| }}

\newcommand {\bra} [1] {\ensuremath{ \left\langle #1 \right| }}
\newcommand {\ket} [1] {\ensuremath{ \left| #1 \right\rangle }}
\newcommand {\ketbratwo} [2] {\ensuremath{ \left| #1 \middle\rangle \middle\langle #2 \right| }}
\newcommand {\ketbra} [1] {\ketbratwo{#1}{#1}}
\newcommand {\braket} [2] {\ensuremath{\left \langle #1 | #2 \right \rangle}}

\newcommand {\prob} [1] {\Fn{\Pr\,}{#1}}

\DeclareMathOperator*{\bigE}{\mathbb{E}}
\newcommand {\expec} [2] {\Fn{\bigE_{\substack{#1}}}{#2}}

\newcommand {\Tr} {\ensuremath{ \mathrm{Tr} }}

\newcommand {\id} {\ensuremath{\mathds{1}}}

\usetikzlibrary{calc}
\tikzset{meter/.append style={draw, inner sep=10, rectangle, font=\vphantom{A}, minimum width=30, line width=.8,
 path picture={\draw[black] ([shift={(.1,.3)}]path picture bounding box.south west) to[bend left=50] ([shift={(-.1,.3)}]path picture bounding box.south east);\draw[black,-latex] ([shift={(0,.1)}]path picture bounding box.south) -- ([shift={(.3,-.1)}]path picture bounding box.north);}}}
 \usetikzlibrary{decorations.pathreplacing}

\NewEnviron{Anonymous}{\ifx\AnonymousSwitch\undefined\BODY\fi}

\graphicspath{ {images/} }

\usepackage{physics}
\usepackage{bm}
\def\f{\frac}
\makeatletter
\DeclarePairedDelimiter{\@lrp}{(}{)}
\DeclarePairedDelimiter{\@lrb}{[}{]}
\DeclarePairedDelimiter{\@lrs}{\{}{\}}
\def\p{\@lrp*}
\def\lrb{\@lrb*}
\def\lrs{\@lrs*}

\def\diag{\mathrm{diag}}
\makeatother
\renewcommand{\tilde}{\widetilde}

\def\x{\textbf{x}}

\newcommand {\Wgt}[2] {\ensuremath{\mathrm{W}^{#1}\br{#2}}}

\newcommand {\AC}[1][] {\ensuremath{\mathbf{AC}^{#1}}}

\newcommand {\QAC}[1][] {\ensuremath{\mathbf{QAC}^{#1}}}
\newcommand {\QACz} {\QAC[0]}

\newcommand {\QLCz} {\ensuremath{\mathbf{QLC}^0}}

\newcommand {\LCz} {\ensuremath{\mathbf{LC}^0}}
\newcommand {\stateQLCz} {\ensuremath{\mathbf{stateQLC}^0}}
\newcommand {\stateQACz} {\ensuremath{\mathbf{stateQAC}^0}}
\newcommand {\CZGate} {CZ-gate}
\newcommand {\CZG} {\ensuremath{\operatorname{CZ}}}
\newcommand {\poly} {\ensuremath{\operatorname{poly}}}
\newcommand {\negl} {\ensuremath{\operatorname{negl}}}

\newcommand {\degeps}[2] {\ensuremath{\tilde{\operatorname{deg}}_{#1}\br{#2}}}

\newcommand {\Parity}[1] {\ensuremath{\operatorname{Parity}_{#1}}}
\newcommand {\Parityn} {\Parity{n}}
\newcommand {\Majorityn} {\ensuremath{\operatorname{Majority}_n}}
\newcommand {\Modraw}[1] {\ensuremath{\operatorname{Mod}_{#1}}}
\newcommand {\Modnk} {\Modraw{n, k}}

\newcommand {\CE} {\ensuremath{\mathcal{E}}}

\newcommand {\EPR} {\text{EPR}}

\newcommand {\Base}[1] {\ensuremath{\mathcal{B}_{#1}}}

\newcommand {\ancillas} {\ensuremath{\psi}}

\newcommand{\parity} {\ensuremath{\text{PARITY}}}

\newcommand{\maj} {\ensuremath{\text{MAJORITY}}}

\newcommand{\modk} {\ensuremath{\text{MOD}_k}}

\newcommand{\cat}[1] {\ket{\Cat_{#1}}}

\newcommand{\circsize}[3] {\ensuremath{\Br{#1\,\middle|\,#2\,\middle|\,#3}}}

\tikzset{
    circnode/.style
    = {rectangle, draw=black, minimum width=2cm, minimum height=0.5cm}
}

\usepackage[english]{babel}

\usepackage[letterpaper,top=2cm,bottom=2cm,left=3cm,right=3cm,marginparwidth=1.75cm]{geometry}

\usepackage{amsmath}
\usepackage{graphicx}

\title{On the Computational Power of \QACz\ with Barely Superlinear Ancillae}

\begin{Anonymous}
\author{
Anurag Anshu\thanks{\scriptsize School of Engineering and Applied Sciences, Harvard University, Cambridge, MA, USA. Email: anuraganshu@fas.harvard.edu.}
\and Yangjing Dong\thanks{\scriptsize State Key Laboratory for Novel Software Technology, New Cornerstone Science Laboratory, Nanjing University, China. Email: dongmassimo@gmail.com.}
\and Fengning Ou\thanks{\scriptsize State Key Laboratory for Novel Software Technology, New Cornerstone Science Laboratory, Nanjing University, China. Email: reverymoon@gmail.com.}
\and Penghui Yao\thanks{\scriptsize State Key Laboratory for Novel Software Technology, New Cornerstone Science Laboratory, Nanjing University, China. Email: phyao1985@gmail.com.}~\thanks{\scriptsize Hefei National Laboratory, Hefei 230088, China.}
}
\end{Anonymous}

\begin{document}
\maketitle

\begin{abstract}

$\QACz$ is the family of constant-depth polynomial-size quantum circuits consisting of arbitrary single qubit unitaries and multi-qubit Toffoli gates.
It was introduced by Moore as a quantum counterpart of $\AC[0]$,
along with the conjecture that $\QACz$ circuits cannot compute \parity.
In this work, we make progress on this long-standing conjecture: we show that any depth-$d$ $\QACz$ circuit requires $n^{1+3^{-d}}$ ancillae to compute a function with approximate degree $\Theta(n)$, which includes \parity, \maj\ and $\mathrm{MOD}_k$. 
 We further establish superlinear lower bounds on quantum state synthesis and quantum channel synthesis. This is the first lower bound on the super-linear sized \QACz. Regarding \parity, we show that any further improvement on the size of ancillae to $n^{1+\exp\br{-o(d)}}$ would imply that $\parity\notin\QACz$.

These lower bounds are derived by giving low-degree approximations to $\QACz$ circuits.
We show that a depth-$d$ $\QACz$ circuit with $a$ ancillae,
when applied to low-degree operators,
has a degree $(n+a)^{1-3^{-d}}$ polynomial approximation in the spectral norm.
This implies that the class $\mathbf{QLC}^0$,
corresponding to linear size $\QACz$ circuits,
has an approximate degree $o(n)$.
This is a quantum generalization of the result that $\mathbf{LC}^0$ circuits have an approximate degree $o(n)$ by
Bun, Kothari, and Thaler.
Our result also implies that $\QLCz\neq\mathbf{NC}^1$.

\end{abstract}

\newpage
\tableofcontents

\newpage

\section{Introduction}




Shallow quantum circuits (quantum circuits with constant depths) stand as a fundamental construction within the realm of quantum computing, arising naturally from the practical constraints and design of quantum hardware. These circuits, while seemingly simple in structure, represent an exciting area of research within quantum information theory and quantum complexity theory. One of the notable achievements in recent years has been the discovery of the quantum advantages of constant-depth quantum circuits~\cite{doi:10.1126/science.aar3106}. It is now well recognized that their computational reach is particularly constrained in their ability to generate long-range quantum entanglement, a cornerstone quantum phenomenon. In essence, the flow of information from a qubit is limited to the ``light cone'' of the qubit in such circuits, severely restricting the computational power of constant-depth quantum circuits, especially when solving decision problems. This raises a fundamental question: What is the simplest architecture beyond shallow quantum circuits that is beyond the light-cone constraint?

Inspired by the study of \AC[0] circuits as extensions of $\mathbf{NC}^0$ circuits, Green, Homer, Moore and Pollett~\cite{moore1999quantum,moor} introduced the notion of $\QACz$ circuits, where they augmented shallow quantum circuits with generalized Toffoli gates, which are defined as
$$U_F\ket{x_1,\ldots,x_n,b}=\ket{x_1,\ldots,x_n,b\oplus\wedge_i x_i}.$$
$\QACz$ gives a simple theoretical model beyond light-cone constraints, with which to study the power of many-qubit operations in quantum computing. Researchers have also investigated the realization of many-qubit operations on different platforms~\cite{9605276,PhysRevApplied.18.034072,Bluvstein2022,PhysRevResearch.4.L042016}.



Despite being a quantum generalization of $\AC[0]$, our understanding of the computational power of $\QACz$ is still very limited. We even do not know whether $\QACz$ contains $\AC[0]$. To see it, arbitrary fan-out is allowed in $\AC[0]$ circuits, which means that the output of a gate can be fed to arbitrary many other gates as input.
However, for quantum circuits, due to the quantum no-cloning theorem, the output of a quantum gate can only be used once as the input for another quantum gate, so fan-out is restricted to $1$.
Another difference that arises from the quantumness of $\QACz$ circuits is that quantum computations are reversible,
while classical circuits can perform the \textbf{AND} and \textbf{OR} gates, which are nonreversible.
To mitigate this, $\QACz$ circuits allow other auxiliary quantum qubits in the circuit, referred to as ancillae. The ancillae act as the memory or internal nodes in classical $\AC[0]$ circuits. The size of the ancillae together with the inputs is analogous to the size of classical circuits.

Does $\QACz$ contain a decision problem that is not in $\AC[0]$?
This problem remains open.
It is well known that $\parity$ and $\maj$ are not in $\AC[0]$~\cite{10.1145/12130.12132}.
The question of whether these functions are contained in $\QACz$ was immediately raised by Moore~\cite{moore1999quantum} when the class $\QACz$ was first proposed.
Since then,
many works have been devoted to investigating this problem, which are largely summarized in \Cref{fig:main-results-function}. However, progress on this problem is still very slow. Before this work, it was still unknown whether any linear-sized \QACz circuit can compute \parity.

Understanding the computational power of bounded-depth circuits is not only interesting in quantum computing but is also a core topic in classical complexity theory.
The celebrated results of H\aa stad showed that $\parity$ is not in $\AC[0]$ via the well-known switching lemma~\cite{10.1145/12130.12132}, which nowadays is still one of the most powerful tools to prove the circuit lower bounds. Despite decades of effort, constant-depth circuits still represent the frontier of our understanding of circuit lower bounds.

A natural approach to proving the lower bounds in $\QACz$ is to generalize the existing techniques for the lower bounds of classical circuits~\cite{10.1145/12130.12132,Razborov1987-fh,10.1145/28395.28404,rossman2016parity}.
However, it seems that all of the techniques face certain barriers when generalizing to the quantum world. As mentioned above, a crucial difference between $\AC[0]$ and $\QACz$ is that all the circuits in $\QACz$ are reversible and, thus the number of output qubits is the same as the number of input qubits. The ancilla qubits act as intermediate nodes in $\AC[0]$ circuits, which play a crucial role in quantum computing. Indeed, Rosenthal proved that an exponential-sized and constant-depth quantum circuit allowing many-qubits Toffoli gates can compute parity~\cite{rosenthal:LIPIcs.ITCS.2021.32}. To our knowledge, all the lower bound techniques for classical circuits, roughly speaking, substitute part of the inputs with simpler ones, which eliminate intermediate nodes. However, the input qubits and the ancilla qubits are fed into the circuit at the same time. Therefore, it is unclear how to eliminate ancilla qubits via simplifying the inputs.

\subsection{Our Results}

Our work extends the Pauli analysis framework on quantum circuits by~\cite{nadimpalli_pauli_2024}.
The main technical result is a low-degree approximation for $\QACz$ circuits:
For an $n$-qubit unitary $U$ implemented by a $\QACz$ circuit, and any low-degree measurement operator $A$,
we show that the resulting Heisenberg-evolved measurement operator $U^\dagger A U$ can not have a $\Omega(n)$ full approximate Pauli degree.
Here an operator having approximate Pauli degree $k$ means that it can be approximated by a summation of $k$-local operators,
and the approximation is with respect to the spectral norm.
The notion of Pauli degree generalizes the Fourier degree of Boolean functions.
In particular, a diagonal operator can be interpreted as the truth table of a Boolean function,
and in this case the Pauli degree and Fourier degree coincide.
The following is the key technical theorem in this paper.

\begin{theorem}[informal of \cref{cor:qac0-whole}]\label{thm:informal:qac0-whole}
  For any $2^n\times2^n$ operator $A$ with degree $\ell$,
  and any unitary $U$ implemented by a depth-$d$ $\QACz$ circuit,
  the approximate degree of $UAU^\dagger$ is upper bounded by $\tilde{O}\br{n^{1-3^{-d}}\ell^{3^{-d}}}$.
\end{theorem}

Since we are working with the spectral norm,
this degree upper bound also holds when we post-select on ancillae.
This is because for an operator $A$ and ancillae in the state $\ket{\varphi}$,
it holds that
\begin{equation*}
    \norm{\br{\id\otimes\bra{\varphi}}A\br{\id\otimes\ket{\varphi}}} \le \norm{A}.
\end{equation*}
In this case, the upper bound of degrees is with respect to the total number of qubits.
For linear ancillae $\QACz$ circuits,
which form the class $\mathbf{QLC}^0$,
\cref{thm:informal:qac0-whole} shows that $\mathbf{QLC}^0$ circuits have approximate degree $o(n)$.
See \cref{cor:approxdegqlc}.

With~\Cref{thm:informal:qac0-whole}, we study the power of $\QACz$ circuits in different quantum computational tasks.

\paragraph{Compute Boolean functions}
By using the relation between Pauli degree and Fourier degree,
our first results are about the hardness for computing Boolean functions. We investigate both the worst case and the average case. In the worst case, the circuit computes all inputs correctly with high probability. For the average case, the circuit computes the function correctly with high probability when the input is drawn uniformly. The results along with previous works are summarized in \Cref{fig:main-results-function}.

\begin{table*}
\centering
\caption{Hardness of Boolean functions}
\label{fig:main-results-function}
\begin{tabular}{|c|c|c|c|}
    \hline
    \multirow{6}{*}{\parity} & $a\geq n2^{-d}-1$ & exact & \cite{10.5555/2011679.2011682}   \\
    \cline{2-4}
    & impossible when $d=2$ & exact & \cite{DBLP:journals/corr/abs-2005-12169}\\
    \cline{2-4}
    & impossible when $d=2$ & average case & \cite{rosenthal:LIPIcs.ITCS.2021.32} \\
    \cline{2-4}
    &$a \le \exp\br{O(n\log n/\varepsilon)}$ when $d=7$ & worst case & \cite{rosenthal:LIPIcs.ITCS.2021.32} \\
    \cline{2-4}
    &$a\geq n^{\Omega(1/d)}$ & average case & \cite{nadimpalli_pauli_2024} \\
    \cline {2-4}
    & $a\geq n^{1+3^{-d}}$ & average/worst case & \textbf{This work}\\
    \hline
    \multirow{2}{*}{\maj} & $a \ge n^{\Omega(1/d)}$ & average case\textsuperscript{$\dagger$} & \cite{nadimpalli_pauli_2024}\\
    \cline{2-4}
    & $a \ge n^{1+3^{-d}}$ & average/worst case & \textbf{This work}\\
    \hline
    \modk & $a \ge n^{1+3^{-d}}$ & worst case & \textbf{This work}\\
    \hline
\end{tabular}

\small\textsuperscript{$\dagger$} For the \maj\ function, the ancillae lower bound for average case works if the average case error is small then $1/\sqrt{n}$. This also applies to our average case lower bound for $\maj$.
\end{table*}

We prove that $\QACz$ circuits with ancillae only slightly more than linear can not compute Boolean functions
with approximate degree $\Omega(n)$ in the worst case with any probability strictly larger than $1/2$.
To the best of our knowledge, these are the first hardness results that allow superlinear pure state ancillae.
\begin{theorem}[informal of \cref{thm:main:WorstCase}]
  Let $f: \set{0,1}^n\to\set{0,1}$ be a Boolean function with approximate degree $\Omega(n)$.
  Suppose $U$ is a depth $d$ $\QACz$ circuit with
  $n$ input qubits and $a = \tilde{O}\br{n^{1+3^{-d}}}$ ancillae initialized in any quantum state.
  Then $U$ cannot compute $f$ with the worst-case error strictly below $1/2$.
  In particular, this includes \parity, \maj, and also \modk\ when $k\le cn$ for some $c < 1$.
\end{theorem}

We further adapt the techniques from~\cite{nadimpalli_pauli_2024} to obtain average-case hardness results.
For Boolean functions that remain high degree even with Frobenius norm approximations,
we prove that it is hard for $\QACz$ circuits to approximate them within a certain error regime.
In particular, we state our average-case hardness results for \parity\ and \maj. 
\begin{theorem}[informal of \cref{thm:main:AverageCase}]
    Suppose $U$ is a depth $d$ $\QACz$ circuit with
    $n$ input qubits and $a = \tilde{O}\br{n^{1+3^{-d}}}$ ancillae initialized in any quantum state.
    Let $C_U(x) \in \set{0,1}$ denote the classical output of the circuit on input $x$. Then it hold that:
    \begin{itemize}
        \item $U$ can not approximate $\Parityn$ over uniform inputs, i.e.,
        \begin{equation*}
            \expec{\x\in\set{0,1}^n}{\prob{C_U(\x) = \Parityn(\x)}} \le \frac{1}{2} + O(d/n).
        \end{equation*}
        \item $U$ can not approximate $\Majorityn$ over uniform inputs, i.e., 
        \begin{equation*}
            \expec{\x\in\set{0,1}^n}{\prob{C_U(\x) = \Majorityn(\x)}} \le 1 - \Omega\br{\frac{1}{\sqrt{n}}} + O(d/n).
        \end{equation*}
    \end{itemize}
   %
\end{theorem}

Our result is incomparable to the results of~\cite{nadimpalli_pauli_2024}:
On the one hand, if the ancillae are initialized into pure states, we have a super-linear lower bound on the ancillae,
while their lower bound is $n^{\Omega\br{1/d}}$.
On the other hand, if the ancillae are initialized as mixed states,
our bound remains $n^{1+3^{-d}}$, while
\cite{nadimpalli_pauli_2024} proved that $\QACz$ with an arbitrarily dimensional maximally mixed state as ancillae cannot compute \parity.

Our lower bound $\tilde{\Omega}\br{n^{1+3^{-d}}}$ seems far away from an arbitrarily polynomial lower bound. Surprisingly, we show that any improvement of the exponent $d$ to $o(d)$ would lead to $\parity\notin\QACz$, a complete resolution of the problem.

\begin{theorem}[informal of \cref{cor:arbitrary}]
        If any $\QACz$ circuit with $n^{1+\exp\br{-o(d)}}$ ancillae,
    where $d$ is the depth of this circuit family,
    can not compute $\Parityn$ with the worst-case error $\negl(n)$,
    then any $\QACz$ circuit family with arbitrary polynomial ancillae can not compute $\Parityn$ with the 
    worst-case error $\negl(n)$.
  %
\end{theorem}
This result is a consequence of the fact that if we can compute \parity\ in $\QACz$,
then for any large enough depth $d$,
there exists a circuit computing $\parity$ using only $n^{1+\exp\br{-d}}$ qubits of ancillae.
The idea is repeating the following simple argument: Note that $\parity$ can be computed recursively. So if we have a $\QACz$ circuit computing $\parity$ with $n^{100}$ ancillae, then we can split the task of $n$-bit $\parity$ into computing $\sqrt{n}$ instances of $\sqrt{n}$-bit $\parity$, and using only $\sqrt{n}\cdot n^{50}$ qubits of ancillae.

\paragraph{Quantum State Synthesis}
We further investigate the hardness of quantum state synthesis in \QACz\ circuits.


\begin{theorem}[informal of \cref{thm:state-synthesis}]
    Let $\varphi = \ketbra{\varphi}$ be a pure state on $n$ qubits
    with approximate degree $\Omega(n)$.
    Suppose that there exists a depth-$d$ $\QACz$ circuit working on $a$ qubits such that the first $n$ qubits of $U\ket{0^a}$ measure to $\varphi$ with constant probability, then we have
    \begin{equation*}
        a = \tilde{\Omega}(n^{1+3^{-d}/2}).
    \end{equation*}
    In particular,
    the $n$-nekomata state,
    and the low energy states of the code Hamiltonian in~\cite{anshu2022nlts} satisfies $\Omega(n)$ approximate degree.
    Hence $\QACz$ circuits with linear ancillae cannot synthesize these states.
\end{theorem}

\paragraph{Quantum Channels Synthesis}
Our last result is a general hardness result for quantum channel synthesis.
\begin{theorem}[informal of \cref{thm:qchannel-degree}] 
    Suppose $\CE_{U, \psi}$ is a quantum channel from $n$ qubits to $k$ qubits,
    implemented by a depth-$d$ \QACz\ circuit $U$ with $n$ input qubits and $a$ ancillae.    
    The upper bound of approximate degree of the Choi representation $\Phi_{U, \psi}$ of  $\CE_{U, \psi}$ is then given by $\tilde{O}\br{(n+a)^{1-3^{-d}}k^{3^{-d}/2}}$.
\end{theorem}


\subsection{Related Works}\label{subsec:relatedwork}


The study of quantum circuit complexity was initiated in large part by Green, Homer, Moore and Pollett~\cite{moore1999quantum,moor},
who introduced quantum counterparts of a number of classical circuit classes, including $\QACz$.
The problem whether the parity function can be computed by \QACz\ circuits was also raised in their papers.
However, this problem has made slow progress over the years.

To compute \parity~, we need to implement unitary
$$\ket{b, x_1, \dots, x_n} \mapsto \ket{b\oplus\bigoplus_ix_i, x_1, \dots, x_n}.$$
Green, Homer, Moore and Pollett~\cite{moor} have shown that \parity is equivalent to fan-out until conjugation with Hadamard gates,
where fan-out is unitary
$\ket{b, x_1, \dots, x_n}\mapsto\ket{b, x_1\oplus b, \dots, x_n\oplus b}$.

Fang, Fenner, Green, Homer, and Zhan~\cite{10.5555/2011679.2011682} gave the first lower bound on this question.
They showed that \QACz circuits with sublinear ancillae cannot compute parity exactly,
where ``exactly'' means that the output qubit should be precisely in the state $\ketbra{0}$ or $\ketbra{1}$ for any input.
In particular, they showed that for any depth-$d$ \QACz\ circuit with $a$ ancillae,
to exactly compute parity, we should have $d\ge 2\log(n/(a+1))$.
This lower bound is derived in the following way:
They are able to find an input state on $(a+1)2^{d/2}$ qubits
such that the circuit output is fixed to be $1$ irrelevant to the other input qubits.
Then if $(a+1)2^{d/2} < n$, such circuits clearly cannot compute parity exactly,
because any flip on another input qubit changes the output.
The case of $0$ ancillae is also implied by a result of Bera~\cite{10.1016/j.ipl.2011.05.002}.
Pad\'e, Fenner, Grier, and Thierau~\cite{DBLP:journals/corr/abs-2005-12169}
showed that no depth-$2$ circuits can
compute the $\Parityn$ function exactly for $n\ge 4$, even with any number of ancillae.
This is achieved by a careful analysis on the structure of depth-$2$ quantum circuits.

These techniques heavily rely on the condition of exact computation, and thus are difficult to extend to the approximate case.
Actually, Rosenthal~\cite{rosenthal:LIPIcs.ITCS.2021.32} has showed a depth-$7$ quantum circuit with many-qubit Toffoli gates and exponentially ancillae that approximately computes the \parity\ function.
This implies that when proving the hardness of computing \parity,
we need to take into account the number of ancilla.
Rosenthal has also established a series of equivalent problems to \parity.
In particular, a circuit preparing the cat state defined as $\ket{\Cat_n} = \frac{1}{\sqrt{2}}\br{\ket{0^n}+\ket{1^n}}$ is equivalent to parity up to constant-degree circuit reductions. In the same work, Rosenthal showed that any depth-$2$ \QACz circuits with arbitrary ancilla qubits cannot compute parity even in the approximate case, generalizing the results of \cite{DBLP:journals/corr/abs-2005-12169} to the average case. Very recently, Nadimpalli, Parham, Vasconcelos, and Yuen~\cite{nadimpalli_pauli_2024} have initiated the study of Pauli spectrum of \QACz\ circuits and proved the first lower bound on $\QACz$ circuit with an arbitrary depth. More specifically, they proved that a \QACz\ circuit with depth $d$ that computes \parity\ needs at least $n^{\Omega(1/d)}$ ancillae.


\paragraph{Fourier Analysis and Pauli Analysis.} Fourier analysis on the space of Boolean functions seeks to understand Boolean functions via their Fourier transformations. The very first application of Fourier analysis on Boolean functions is the well-known KKL theorem proved by Kahn, Kalai, and Linial~\cite{21923} in 1988. Today, Fourier analysis on Boolean functions has played a crucial role in various areas. Readers may refer to O'Donnell's excellent book~\cite{ODonnell2014} for a thorough treatment.

Pauli analysis is a quantum generalization of the Fourier analysis of Boolean functions. Recall that the set of $n$-qubit Pauli operators $\mathcal{P}_n=\set{I,X,Y,Z}^{\otimes n}$ forms an orthonormal basis for the space of all operators acting on $n$-qubits. Given an $n$-qubit operation $M$, its Pauli expansion is $M=\sum_{P\in\mathcal{P}_n}\widehat{M}\br{P}P$, analogous to the Fourier expansion of a Boolean function, where the coefficients $\widehat{M}\br{P}$'s are the Pauli spectrum of $M$. Such expansions can be further generalized to quantum channels~\cite{pmlr-v195-bao23b,nadimpalli_pauli_2024}. To see that Pauli analysis is indeed a generalization of Fourier analysis, for a Boolean function $f:\set{0,1}^n\rightarrow\mathbb{R}$, the Pauli expansion of the diagonal matrix $M_f=\sum_{x\in\set{0,1}^n}f\br{x}\ketbra{x}$ is exactly the same as the Fourier expansion of $f$. 

Pauli analysis was first studied by Montanaro and Obsorne in the context of the so-called quantum Boolean functions~\cite{MO10} and has received increasing attention in the past couple of years, finding applications in testing and learning quantum operations~\cite{chen2024predictingquantumchannelsgeneral,pmlr-v195-bao23b,doi:10.1137/1.9781611977554.ch43,arunachalam_et_al:LIPIcs.ICALP.2024.13,
Rouze2024}, classical simulations of noisy quantum circuits~\cite{
10.1145/3564246.3585234}, multiprover quantum interactive proof systems in the noisy world~\cite{yao2019doublyexponentialupperbound,doi:10.1137/20M134592X,qin_et_al:LIPIcs.ICALP.2023.97,dong_et_al:LIPIcs.CCC.2024.30}. Very recently, Nadimpalli, Parham, Vasconcelos, and Yuen~\cite{nadimpalli_pauli_2024} initiated the study on the Pauli spectrum of \QACz circuits,  which has led to a series of lower bounds and learning algorithms for \QACz.

\paragraph{Comparison with \cite{nadimpalli_pauli_2024}}

Similarly to the work~\cite{nadimpalli_pauli_2024}, our work is also built on the theory of Pauli analysis. Instead of analyzing the Pauli spectrum of the channels induced by \QACz\ circuits as in~\cite{nadimpalli_pauli_2024}, we study the Pauli spectrum of the projectors induced by a \QACz\ circuits, namely $U^{\dagger}\br{\ketbra{0}\otimes \id}U$, where $U$ is a \QACz\ circuit. 

\cite{nadimpalli_pauli_2024} removed the large Toffoli gates in a circuit
and argued that the Frobenius norm does not change significantly.
In this way, they obtain a low-degree approximation of \QACz\ circuits within the Frobenius norm.
However, with this approach, the ancillary qubits are restricted to $n^{O(1/d)}$,
because the normalized Frobenius norm grows exponentially after operation
$A\mapsto\br{\id\otimes\bra{0^a}}A\br{\id\otimes\ket{0^a}}$, which projects the ancillary registers to the initial state $\ket{0^a}$.

In our work, we consider the approximation with respect to the spectral norm. We use Chebyshev polynomials to approximate the high-degree gates, and thus are able to get $\sqrt{n}$-degree approximations with respect to the spectral norm. Meanwhile, the projection of the ancilla registers to the initial state does not increase the spectral norm. 
One more advantage of using the spectral norm is that we can prove the hardness for a much broader class of Boolean functions.
Since the Frobenius norm can be exponentially smaller than the spectral norm,
many Boolean functions that have low-degree approximations with small Frobenius norm
actually have an $\Omega(n)$ approximate degree if we consider the spectral norm.
Hence in our work we are also able to prove optimal $1/2$ worst case hardness results for the \maj\ function and the \modk\ function.
For comparison, \cite{nadimpalli_pauli_2024} shows that it is hard for $\QACz$ circuits
to approximate the \maj\ function with a success probability greater than $1 - \frac{1}{\sqrt{n}}$.

In this paper, we further establish lower bounds on quantum state synthesis. Notice that the normalized Frobenius norms of quantum states are exponentially small,
while the spectral norm can be as large as $1$.
Hence, the results of \cite{nadimpalli_pauli_2024} cannot be applied directly to quantum state synthesis.

\subsection{Discussion}

It is still open whether $\QACz$ circuits can compute the \parity\ function.
Despite being hard in the quantum case,
the classical counterpart, that is, $\AC[0]$ circuits cannot compute the \parity\ function,
has admitted several different proofs.
Here we briefly explain some well-known techniques for classical circuit lower bounds and explain why it seems to be difficult to generalize to quantum circuits.

The most well studied technique goes to H\aa stad's switching lemma~\cite{10.1145/12130.12132},
which describes the effect of random restrictions on depth-$2$ circuits.
Random restriction fixes part of the input in a random manner,
which effectively kills large \textbf{AND} or \textbf{OR} gates. Moreover, \textbf{AND}-\textbf{OR} circuits and \textbf{OR}-\textbf{AND} circuits switch with each other via decision tree representation, which decreases the depth of circuits.
For quantum circuits, however, there is not an obvious way to apply random restrictions.
Fang, Fenner, Green, Homer, and Zhang~\cite{10.5555/2011679.2011682}
proved that for any depth $d$ $\QACz$ circuit with $a$ ancillae,
there exists a quantum state $\varphi$ on $(a+1)2^{d/2}$ qubits such that
when part of the input is restricted to the state $\varphi$,
then the output of the $\QACz$ circuit is fixed to be $1$.
This can be seen as a restriction-based method.
However, in their work the state $\varphi$ is chosen in a deterministic manner.
Also, their method only works if the circuit computes \parity\ exactly.
The difficulty for quantum random restrictions
arises from the fact that large Toffoli gates might interact with many ancilla qubits but only a few input qubits, such as Rosenthal's circuit for nekomata states synthesis~\cite{rosenthal:LIPIcs.ITCS.2021.32}. Then random restricting inputs  won't kill many large Toffoli gates.

Rossman~\cite{rossman2016parity} provided a simple proof of $\mathrm{PARITY}\notin\AC[0]$ via the DNF sparsification introduced by Gopalan, Meka and Reingold~\cite{6243388}. 
An $m$-junta is a Boolean function $f: \set{0,1}^n\to\mathbb{R}$ that actually only looks at the input
in the $m$ indices.
The DNF sparsification theorem from \cite{6243388} allows one to approximate any $k$-DNF or $k$-CNF to an
$m$-junta, where $m$ only depends on $k$ and the approximation error.
Then, through a random restriction, one can convert a large junta to a smaller junta.
These two steps together can convert an depth-$3$ circuit,
which are essentially $k$-DNFs or $k$-CNFs, into a $k$-junta,
which can be represented by a constant sized depth-$2$ circuit,
decreasing the depth by $1$.
By repeating the process above,
we can reduce an arbitrary $\AC[0]$ circuit into a trivial depth-$2$ $k$-junta.
On the other hand, the \parity\ function is robust against random restrictions,
thus we can prove that $\AC[0]$ circuits cannot compute the \parity\ function.

This methods appears to be easier to adapt to the quantum case,
compared to the switching lemma of~\cite{10.1145/12130.12132}.
This is because the notion of ``quantum junta'' is more natural and has been studied in~\cite{MO10,Wang2011PropertyTO,doi:10.1137/1.9781611977554.ch43,pmlr-v195-bao23b}.
However, it is still built on random restrictions, and it is not clear how to extend it to quantum circuits.

The polynomial method is a different approach to prove the hardness results of $\AC[0]$.  Razborov~\cite{Razborov1987-fh} and Smolensky~\cite{10.1145/28395.28404} to prove that $\maj$ and $\mathrm{MOD}_3$ cannot be efficiently computed by $\mathbf{AC}^0[\oplus]$.
Refer to the excellent survey of polynomial methods by Williams~\cite{williams:LIPIcs.FSTTCS.2014.47} for the polynomial method in circuit complexity.
For classical circuits, it is first shown that the elementary gates (\textbf{AND}, \textbf{OR}, and \textbf{NOR})
can be approximated by low-degree polynomials.
And then by combining these polynomial together,
we can represent the whole $\AC[0]$ circuit by a low-degree polynomial.
Finally, since \parity, \maj, and \modk\ have an $\Omega(n)$ approximate degree,
we can show that these problems are hard for $\AC[0]$ circuits.

Our techniques resemble the polynomial method in spirit,
where low-degree polynomials are replaced by low-degree operators.
The elementary gates in a $\QACz$ circuit are the multi-qubit \CZGate s,
along with all single-qubit unitaries.
The single-qubit unitaries are already of degree $1$ so we do not need further actions on them.
For an $n$ qubit \CZGate, we use \cref{cor:approximate-CZGate} to represent it by a degree $\tilde{O}(\sqrt{n})$ operator.
Using a layer-by-layer argument,
we prove \cref{cor:qac0-whole},
which states that the whole effect of a $\QACz$ circuit is low degree. A crucial difference between our proof and the proofs for \AC[0] is that the approximation for \AC[0]\ is over a certain finite field while our approximation for \QACz\ is over complex number fields.
It is not clear how to represent the computation of \QACz\ in a finite field. 



Our work uses spectral norm approximations, which do not increase when the number of ancillae increase.
However, this requires us to approximate $\QACz$ circuits in the spectral norm using low-degree operators.
Even in the classical case, proving that $\AC[0]$ circuits have $o(n)$ approximate degree is a notorious open problem~\cite{10.1145/3444815.3444825}, and if solved we may get several consequences in complexity theory;
see \cite{slote:LIPIcs.ITCS.2024.92} and \cite{10.1145/3444815.3444825} for more discussion.
Our work can be seen as a quantum generalization of the work of Bun, Kothari, and Thaler~\cite[Theorem 5]{doi:10.1137/1.9781611975482.42},
who proved that $\AC[0]$ circuits with linear ancillae and depth $d$, aka $\mathbf{LC}^0_d$,
have an approximate degree $\tilde{O}(n^{1-2^{-d}})$.
However, in both the quantum case and the classical case,
proving approximate degree upper bounds for $\QACz$ circuits with superlinear ancillae or $\AC[0]$ circuits of superlinear size requires novel techniques.




\begin{Anonymous}
\subsection*{Acknowledgement} 
We thank the anonymous reviewers for their careful reading of our manuscript and their many insightful comments and suggestions.
We thank Zongbo Bao for his contribution in the early stage of this project. PY thanks Xun Gao for the helpful discussion. 
YD, FO, and PY were supported by National Natural Science Foundation of China (Grant No. 62332009, 12347104), Innovation Program for Quantum Science and Technology (Grant No. 2021ZD0302901), NSFC/RGC Joint Research Scheme (Grant no. 12461160276), Natural Science Foundation of Jiangsu Province (No. BK20243060) and the New Cornerstone Science Foundation.
\end{Anonymous}

\section{Preliminaries}




\subsection{Analysis of Boolean Functions}

In this subsection, we briefly introduce the theory of analysis of Boolean functions. Readers may refer to O'Donnell's excellent book~\cite{ODonnell2014} for a thorough treatment.

Given a Boolean function $f: \set{0,1}^n\to\mathbb{R}$,
its $p$-norm is defined to be $\normsub{f}{p} = \br{\expec{\x}{\abs{f(\x)}^p}}^{1/p}$ for $p\geq 1$.
Its infinity norm is defined to be $\normsub{f}{\infty} = \lim_{p\to\infty}\normsub{f}{p} = \max_{x}\abs{f(x)}.$
We will use the notation $\norm{f} = \normsub{f}{\infty}$.

Given Boolean functions $f,g:\set{0,1}^n\to\mathbb{R}$, the inner product of $f$ and $g$ is $\langle f,g\rangle=\expec{\x}{f(\x)g(\x)}$, where $\mathbf{x}$ is uniformly distributed over $\set{0,1}^n$. 
For any $S\subseteq[n]$, define the Fourier basis $\chi_S$ as
$\chi_S(x) = (-1)^{\sum_{i\in S}x_i}.$ It is not hard to see that $\set{\chi_S}_{S\subseteq[n]}$ is an orthonormal basis. The Fourier expansion of $f$ is $f = \sum_{S\subseteq[n]}\widehat{f}(S)\chi_S.$


The following is the well-known Parseval theorem, which relates the $2$-norm and Fourier coefficients of a Boolean function.
\begin{theorem}[Parseval's theorem]\label{thm:Parseval}
    Let $f: \set{0,1}^n\to\mathbb{R}$ be a Boolean function.
    Then
    \begin{equation*}
        \normsub{f}{2}^2 = \sum_{S\subseteq[n]}\widehat{f}(S)^2.
    \end{equation*}
\end{theorem}

\begin{definition}
    Let $f: \set{0,1}^n\to\mathbb{R}$ be a Boolean function with Fourier expansion
    $$f = \sum_{S\subseteq[n]}\widehat{f}(S)\chi_S.$$
    Then the degree of $f$ is defined as
    $$\deg\br{f} = \max_{S: \widehat{f}(S) \neq 0} \abs{S}.$$
\end{definition}

\begin{definition}[Approximate Degree]
    Let $f: \set{0,1}^n\to\mathbb{R}$ be a Boolean function.
    For $\varepsilon\in[0, 1]$, the approximate degree of $f$ is defined as
    $$\degeps{\varepsilon}{f} = \min_{g: \norm{f - g} \le \varepsilon} \deg\br{g}.$$
\end{definition}

It is worth noticing that the approximation is with respect to the infinity norm. The notion of approximate degrees has played a crucial role in quantum query complexity and quantum communication complexity~\cite{10.1145/502090.502097,BurhmanWolf:2001}.

A Boolean function $f:\set{0,1}^n\rightarrow\mathbb{R}$ is symmetric if $f\br{x_1,\ldots,x_n}=f\br{x_{\sigma(1)},\ldots,x_{\sigma(n)}}$ for any $\sigma\in S_n$.

\begin{fact}[{\cite[Theorem 3]{10.1145/3444815.3444825}}]\label{lem:locality-of-function}
    Let $f: \set{0, 1}^n\to\set{0,1}$ be a symmetric function and $t$ be the number such that $f$ is constant on all inputs of Hamming weight between $t$ and $n-t$. Then for $\varepsilon\in\br{2^{-n}, 1/3}$, we have
    $$
        \degeps{\varepsilon}{f} = \Theta\br{\sqrt{nt} + \sqrt{n\log\br{1/\varepsilon}}}.
    $$
\end{fact}

\begin{example}\label{cor:locality-of-parity}
    With \cref{lem:locality-of-function},
    we can show that the following functions have an $\Omega(n)$ approximate degree.
    \begin{itemize}
    \item For any $n$ define the function $\Parityn: \set{0,1}^n\to\set{0,1}$ as
        $$\Parityn(x) = \bigoplus_ix_i.$$
        Then $\degeps{1/3}{\Parityn} = \Theta\br{n}$.
    This result is first derived in~\cite{10.1145/129712.129758}.

    \item
        For any odd $n$ define the function $\Majorityn:\set{0,1}^n\to\set{0,1}$ as
        $$
        \Majorityn(x) = \begin{cases}
            1 &\text{ if } \sum_ix_i\ge n/2\\
            0 &\text{ if } \sum_ix_i< n/2
        \end{cases}
        $$
        Then $\degeps{1/3}{\Majorityn} = \Theta\br{n}$.
    \item
        For any $n$ and $k$, define the function $\Modnk: \set{0,1}^n\to\set{0,1}$ as
        $$
        \Modnk(x) = \begin{cases}
            1 &\text{ if } \sum_ix_i \bmod k \neq 0 \\
            0 &\text{ if } \sum_ix_i \bmod k = 0 \\
        \end{cases}
        $$
        Then $\degeps{1/3}{\Modnk} = \Omega\br{n - k}$.
        Note that $\Modraw{n, 2} = \Parityn$.
  \end{itemize}
\end{example}

\subsection{Quantum Information and Pauli analysis}

A quantum system $A$ is associated with a finite-dimensional Hilbert space, which we also denote by $A$.
The quantum registers in the quantum system $A$ are represented by {\it density operators},
which are trace-one positive semidefinite operators, in the Hilbert space $A$.
We also use the Dirac notation $\ket{\varphi}$ to represent a pure state.
In this case, we have the convention that $\varphi = \ketbra{\varphi}$,
where here $\varphi$ is a rank-one density operator.
For two separate quantum registers $\varphi$ and $\sigma$ from quantum systems $A$ and $B$,
the compound register is the Kronecker product $\varphi\otimes\sigma$.
A {\it positive operator-valued measure} (POVM) is a quantum measurement described by a set of positive semidefinite operators that sum up to identity.
Let $\set{P_a}_{a}$ be a POVM applied on a quantum register $\varphi$,
then the probability that the measurement outcome is $a$ is $\Tr\Br{P_a\varphi}$.

For any integer $n\geq 2$, let $\mathcal{M}_n$ be the set of $n\times n$ matrix.
For any matrix $M\in\mathcal{M}_n$, we let $\abs{M} = \sqrt{M^\dagger M}$. For any $M,N\in\mathcal{M}_n$, the inner product of $M, N$ is $\langle M,N\rangle=\Tr\Br{M^{\dagger}N}/n$. It is evident that $\br{\mathcal{M}_n,\langle\cdot,\cdot\rangle}$ forms a Hilbert space.

For $p\ge 1$, the {\it normalized} Schatten $p$-norm of $M$ is defined to be
\begin{equation*}
    \normsub{M}{p} = \br{\frac{1}{n}\Tr\Br{\abs{M}^p}}^{1/p}.
\end{equation*}
It is not hard to see that $\langle M,M\rangle=\normsub{M}{2}^2$. Moreover, $\normsub{\cdot}{p}$ is monotone non-decreasing with respect to $p$ and $\normsub{\cdot}{\infty}=\lim_{p\rightarrow\infty}\normsub{\cdot}{p}$ is the spectral norm.
The fidelity between two quantum states $\rho$ and $\varphi$ is defined as
\begin{equation*}
    F(\rho, \sigma) = \Tr\Br{\sqrt{\sqrt{\rho}\sigma\sqrt{\rho}}}.
\end{equation*}
The above definition is symmetric: $F(\rho, \sigma) = F(\sigma, \rho)$.
When one of the input states is pure, say $\rho = \ketbra{\rho}$,
then we have
\begin{equation*}
    F(\ketbra{\rho}, \sigma) = \sqrt{\bra{\rho}\sigma\ket{\rho}}.
\end{equation*}
The Fuchs–van de Graaf inequalities give a relation between the norms and fidelity:
\begin{lemma}[{\cite[Theorem 3.33]{watrous2018theory}}]\label{lem:fuchs-vandegraaf}
    Let $\rho, \sigma$ be positive semi-definite operators of size $2^n\times 2^n$.
    Let $\normsub{\rho}{\text{TD}} = 2^n\normsub{\rho}{1}$ denote the unnormalized trace norm of an operator.
    It holds that
    \begin{equation*}
        1 - \frac{1}{2}\normsub{\rho - \sigma}{\text{TD}} \le F(\rho, \sigma) \le \sqrt{1 - \frac{1}{4}\normsub{\rho - \sigma}{\text{TD}}^2}.
    \end{equation*}
    Equivalently,
    \begin{equation*}
        2 - 2F(\rho, \sigma) \le \normsub{\rho - \sigma}{\text{TD}} \le 2\sqrt{1 - F(\rho, \sigma)^2}.
    \end{equation*}
    Also, for any operator $\rho$, we have $\normsub{\rho}{\text{TD}} \ge \normsub{\rho}{p}$ for any $p\ge 1$ or $p=\infty$.
\end{lemma}

The following lemma will also be used throughout this work.
The proof is deferred to \cref{appendix:proofs}.
\begin{lemma}\label{lem:spectral-multiply}
    Let $A, B, \tilde{A}, \tilde{B}$ be operators satisfying
    \begin{itemize}
        \item $\norm{A} \le 1$ and $\norm{B} \le 1$.
        \item $\norm{A - \tilde{A}} \le \varepsilon_0$.
        \item $\norm{B - \tilde{B}} \le \varepsilon_1$.
    \end{itemize}
    Then $\norm{AB} \le 1$ and
    $\norm{AB - \tilde{A}\tilde{B}} \le \varepsilon = \varepsilon_0 + \varepsilon_1 + \varepsilon_0\varepsilon_1 = (1+\varepsilon_0)(1+\varepsilon_1) - 1$.
\end{lemma}

The Pauli matrices $\Base{0}, \dots, \Base{3}$ are
$$
\Base{0} = \begin{bmatrix}
    1 & 0 \\
    0 & 1
\end{bmatrix},
\Base{1} = \begin{bmatrix}
    0 & 1 \\
    1 & 0
\end{bmatrix},
\Base{2} = \begin{bmatrix}
    0 & -i \\
    i & 0
\end{bmatrix},
\Base{3} = \begin{bmatrix}
    1 & 0 \\
    0 & -1 
\end{bmatrix},
$$
which form an orthonormal basis in $\mathcal{M}_2$. For integer $n\geq 1$ and $\sigma\in\set{0,1,2,3}^n$, we define 
$$\Base{\sigma} = \Base{\sigma_1}\otimes\dots\otimes\Base{\sigma_n}.$$
The set of Pauli matrices $\set{\Base{\sigma}}_{\sigma\in\set{0,1,2,3}^n}$ forms an orthonormal basis in $\mathcal{M}_{2^n}$.
For any $2^n\times2^n$ matrix $A$,
the Pauli expansion of $A$ is
$$A = \sum_{\sigma\in\set{0,1,2,3}^n}\widehat{A}(\sigma)\cdot\Base{\sigma}.$$
The coefficients $\widehat{A}(\sigma)$'s are called the Pauli coefficients of $A$.
We can then define the degree and the approximate degree of a matrix.
\begin{definition}\label{def:quantum-approximate-degree}
    Let $n$ be an integer and $A$ be a $2^n\times2^n$ matrix.
    The degree of $A$ is defined as
    $$\deg(A) = \max_{\sigma: \widehat{A}(\sigma) \neq 0} \abs{\sigma},$$
    where $\abs{\sigma} = \abs{\set{i: \sigma_i\neq 0}}$.
    For $\varepsilon\in[0, 1]$, the approximate degree of $A$ is defined as
    $$\degeps{\varepsilon}{A} = \min_{B: \norm{A - B} \le \varepsilon} \deg\br{B},$$
    where $\norm{\cdot}$ is the spectral norm.
\end{definition}


\begin{lemma}\label{lem:approximate-local-U-tensor}
    Let $A$ be a $2^n\times2^n$ matrix satisfying $\deg\br{A} = \ell$.
    Let $U$ be a unitary of the form $U = \bigotimes_iU_i$,
    where each $U_i$ is a local unitary acting on at most $t$ qubits.
    Then $$\deg\br{UAU^\dagger} \le \ell t.$$
\end{lemma}

\begin{proof}
    Since $A$ has degree at most $\ell$, we have the Pauli expansion
    $$A = \sum_{\substack{\sigma\in\set{0,1,2,3}^n\\\abs{\sigma}\le\ell}}\widehat{A}(\sigma)\cdot\Base{\sigma}.$$
    Then
    $$UAU^\dagger = \sum_{\substack{\sigma\in\set{0,1,2,3}^n\\\abs{\sigma}\le\ell}}\widehat{A}(\sigma)\cdot U\Base{\sigma}U^\dagger.$$
    Note that each $U\Base{\sigma}U^\dagger$ acts non-trivially on at most $\ell t$ qubits,
    because for the $U_j$ that do not act on the non-trivial part of $\Base{\sigma}$, we have $U_jU_j^\dagger = \id$.
    This proves that $\deg\br{UAU^\dagger} \le \ell t$.
\end{proof}

\begin{lemma}\label{lem:diagonalFourier}
    Let $M$ be a $2^n\times2^n$ matrix and $\diag(M)$ be the diagonal matrix containing the diagonal entries of $M$.
    If $M$ has a Pauli expansion
    $$M = \sum_{\sigma\in\set{0,1,2,3}^n}\widehat{M}(\sigma)\Base{\sigma},$$
    then
    $$\diag(M) = \sum_{\sigma\in\set{0,3}^n}\widehat{M}(\sigma)\Base{\sigma}.$$
\end{lemma}
\begin{proof}
    $$\diag(M) = \sum_x\bra{x}M\ket{x}\cdot\ketbra{x} = \sum_{x, \sigma}\widehat{M}(\sigma)\cdot\bra{x}\Base{\sigma}\ket{x}\cdot\ketbra{x}.$$
    Since $\Base{1} = \begin{bmatrix} 0 & 1 \\ 1 & 0\end{bmatrix}$,
    we have $\bra{b}\Base{1}\ket{b} = 0$ for any $b\in\set{0,1}$.
    This also holds for $\Base{2}$.
    So $\bra{x}\Base{\sigma}\ket{x} \neq 0$ only if $\sigma\in\set{0,3}^n$.
\end{proof}

One consequence is that the approximate degree of the diagonal part does not exceed the approximate degree of the original matrix.
\begin{lemma}\label{lem:diagonaldegree}
  Given a $2^n\times2^n$ matrix $M$, it holds that, for any $\varepsilon$,
  $$\degeps{\varepsilon}{\diag(M)} \le \degeps{\varepsilon}{M}.$$
\end{lemma}

\begin{proof}
    Let $\tilde{M}$ be the operator satisfying
    \begin{itemize}
        \item $\norm{M-\tilde{M}} \le \varepsilon$.
        \item $\deg(\tilde{M}) = \degeps{\varepsilon}{M}$.
    \end{itemize}
    By \cref{lem:diagonalFourier},
    $\deg\br{\diag\br{\tilde{M}}} \le \deg\br{\tilde{M}}$.
    Also
    \begin{align*}
    \norm{\diag(M) - \diag\br{\tilde{M}}} &= \max_{x}\abs{\bra{x}\br{\diag(M) - \diag\br{\tilde{M}}}\ket{x}} \\
      &= \max_{x}\abs{\bra{x}\br{M - \tilde{M}}\ket{x}} \\
      &\le \norm{M - \tilde{M}}.
    \end{align*}
\end{proof}

\begin{lemma}\label{lem:zeropart}
    Let $M$ be a $2^n\times2^n$ Hermitian matrix.
    For any $k\le n$ and  $2^k\times 2^k$ density operator $\varphi$, let
    \begin{equation*}
        M_{\varphi} = \Tr_{n-k+1, \dots, n}\Br{\br{\id\otimes\varphi} M}.
    \end{equation*}
    For any $\varepsilon\in[0,1]$, we have
    \begin{equation*}
        \degeps{\varepsilon}{M_{\varphi}} \le \degeps{\varepsilon}{M}.
    \end{equation*}
\end{lemma}
\begin{proof}
    We first consider the exact degree, i.e., $\varepsilon = 0$.
    Suppose $M$ has Fourier expansion
    \begin{equation*}
        M = \sum_{\sigma\in\set{0,1,2,3}^{n}} \widehat{M}(\sigma) \Base{\sigma}.
    \end{equation*}
    Then
    \begin{align*}
        M_\varphi &= \Tr_{n-k+1, \dots, n}\Br{\br{\id\otimes\varphi} M} \\
                &= \sum_{\sigma\in\set{0,1,2,3}^{n}} \widehat{M}(\sigma) \Tr_{n-k+1, \dots, n}\Br{\br{\id\otimes\varphi} \Base{\sigma}} \\
                &= \sum_{\substack{\sigma_1\in\set{0,1,2,3}^{n-k}\\\sigma_2\in\set{0,1,2,3}^{k}}} \widehat{M}(\sigma_1\sigma_2) \Tr\Br{\varphi \Base{\sigma_2}}\Base{\sigma_1} \\
    \end{align*}
    and we can easily see that $\deg(M_{\varphi}) \le \deg(M)$.
    
    Then for $\varepsilon > 0$,
    Let $\tilde{M}$ be the matrix satisfying $\deg(\tilde{M}) = \degeps{\varepsilon}{M}$ and $\norm{M-\tilde{M}} \le \varepsilon$.
    Then let $\tilde{M}_{\varphi} = \Tr_{k+1, \dots, n}\Br{\br{\id\otimes\varphi} \tilde{M}}$.
    From the $\varepsilon = 0$ result We have $\deg(\tilde{M}_{\varphi}) \le \deg(\tilde{M})$.
    Let $S_n=\set{\phi\in\mathcal{M}_{2^n}: \Tr\abs{\phi}=1}$. Then
    \begin{align*}
        \norm{M_{\varphi} - \tilde{M}_{\varphi}}
            &= \max_{\phi\in S_{n-k}} \Tr\Br{\br{M_{\varphi} - \tilde{M}_{\varphi}}\phi} \\
            &= \max_{\phi\in S_{n-k}} \Tr\Br{\br{M - \tilde{M}}\br{\phi\otimes\varphi}} \\
            &\le \max_{\phi\in S_n} \Tr\Br{\br{M - \tilde{M}} \phi} \\
            &= \norm{M - \tilde{M}} \\
            &\le \varepsilon.
    \end{align*}
\end{proof}

\subsection{Embedding Boolean functions in the space of operators}\label{section:pre:functionsandoperators}

A natural approach to embedding Boolean functions in the space of matrices is viewing a Boolean function as a diagonal matrix. More specifically, let $f: \set{0,1}^n\to\mathbb{R}$ be a function, we use $M_f$ to denote the $2^n\times2^n$ square matrix
\begin{equation}\label{eqn:diagonalf}
M_f = \sum_xf(x)\cdot\ketbra{x} = \begin{bmatrix}f(0^n) & & \\ & \ddots & \\  & & f(1^n)\end{bmatrix}.    
\end{equation}

The notion of Fourier expansion of $f$ coincides with the Pauli expansion of $M_f$:
\begin{fact}
    Suppose $f$ has a Fourier expansion
    $$f = \sum_{S\subseteq[n]}\widehat{f}(S)\chi_S.$$
    For each $\sigma\in\set{0,3}^n$, define $S_\sigma = \set{i: \sigma_i = 3}$.
    Then $M_f$ has Pauli expansion
    $$M_f = \sum_{\sigma\in\set{0,3}^n}\widehat{f}(S_\sigma)\cdot\Base{\sigma}.$$
\end{fact}
One consequence is that the degrees and the approximate degrees of $f$ and $M_f$ are the same, respectively.
\begin{fact}\label{fact:degreecoincide}
    $\degeps{\varepsilon}{f} = \degeps{\varepsilon}{M_f}$ for any $\varepsilon\in[0,1]$.
\end{fact}

\begin{proof}
    Note that since $M_f$ is diagonal, by \cref{lem:diagonalFourier},
    the operator $B$ that satisfies $\norm{M_f - B}\le\varepsilon$ and minimizes the degree is also diagonal.
\end{proof}

\subsection{Quantum Circuits}

A quantum circuit is a model of quantum computation.
The computation involves an initial input quantum register with the state $\ket{\varphi}$,
an ancillary quantum register with the state $\ket{\ancillas}$,
and a series of quantum gates $U_s, \dots, U_1$,
where each $U_i$ is a unitary operator drawn from a predefined gate set $\mathcal{U}$,
and acts on a subset of working qubits.
After the computation, the working quantum register contains the state
$$U\br{\ket{\varphi}\otimes\ket{\ancillas}} = U_s\dots U_1\br{\ket{\varphi}\otimes\ket{\ancillas}}.$$

We can implement an $n$-qubit to $k$-qubit quantum channel
by tracing out all but the first $k$ qubits after implementing all unitaries.
This channel with ancillae in the state $\ket{\ancillas}$ is denoted by $\Phi_{k, U, \ket{\ancillas}}$. The subscript $k$ is omitted whenever it is clear from the context.

We may get a classical output by applying a computational basis measurement on the first qubit.
That is, we apply the measurement $\set{M_0 = \ketbra{0}\otimes\id, M_1 = \ketbra{1}\otimes\id}$,
and the probability that we output $1$ is 
\begin{equation*}
    \Tr\Br{\br{\ketbra{1}\otimes\id}U\br{\ketbra{\varphi}\otimes\ketbra{\ancillas}}U^\dagger}.
\end{equation*}
We use $C_{U, \ket{\ancillas}}$ to denote the above classical output of a quantum channel.
When $\ket{\ancillas} = \ket{0}^a$ or there are no ancillae, we may simply write $C_U$.


We say a quantum circuit $C$ computes $f$ with the worst-case probability $1-\varepsilon$ (with the worst-case error $\varepsilon$) if  for any $x\in\set{0,1}^n$,
  \begin{equation*}
    \prob{C(x) \neq f(x)} \le \varepsilon,
  \end{equation*}
  where $C(x)$ is the output of the circuit on input $x$.  
  Similarly, a quantum circuit computes $f$ with the average-case probability $1-\varepsilon$ (with the average-case error $\varepsilon$) if
  \begin{equation*}
    \expec{\x\in\set{0,1}}{\prob{C(\x) \neq f(\x)}} \le \varepsilon.
  \end{equation*} 

In this work we are concerned with \QAC\ circuits.
The gate set for \QAC\ circuits includes all single-qubit unitaries,
along with the multi-qubit \CZGate s\footnote{Some definitions use Generalized Toffoli gates. They are equivalent in our case.}.
An $n$-qubit \CZGate\ is defined as
\begin{equation}
    \CZG = \id - 2\ketbra{1}^n.
\end{equation}
The gates can be written in the form $U = L_{d}M_{d}\dots M_1L_0$,
where each $L_i$ is a tensor product of single qubit unitaries,
and each $M_i$ is a tensor product of \CZGate s.
The depth of this circuit is $d$.
A sample $\QAC$ circuit with input state $\ket{x}$ and ancillae $\ket{\ancillas}$ is depicted in \cref{fig:qac-circuit}.

\begin{figure}
    \begin{equation*}
    \Qcircuit @C=1em @R=.7em {
       \lstick{} & \gate{U_1} & \multigate{3}{\CZG_1} & \qw & \multigate{1}{\CZG_2} & \gate{U_7} & \meter \\
       \lstick{} & \gate{U_2} & \ghost{\CZG_1} & \gate{U_5} & \ghost{\CZG_2} & \gate{U_8} & \qw
       \inputgroupv{1}{2}{.9em}{.9em}{\ket{x}} \\
       \lstick{} & \gate{U_3} & \ghost{\CZG_1} & \gate{U_6} & \multigate{1}{\CZG_3} & \qw & \qw \\
       \lstick{} & \gate{U_4} & \ghost{\CZG_1} & \qw & \ghost{\CZG_3} & \qw & \qw
       \inputgroupv{3}{4}{.9em}{.9em}{\ket{\ancillas}} \\
       \lstick{} & {L_0} & {M_1} & {L_1} & {M_2} & {L_2} &
    }
    \end{equation*}
    \Description{An example of a QAC0 circuit composing of some random single qubit unitaries and CZ Gates.}
    \caption{\QACz\ Circuit Example}
    \label{fig:qac-circuit}
\end{figure}

The complexity class $\QACz$ consists of all languages that can be decided by constant-deph and polynomial-sized \QAC\ quantum circuits. Formally, a language $L$ is in $\QACz$ if there exists a family of constant-deph and polynomial-sized \QAC\ quantum circuits $\set{C_n}_{n\in\mathbb{N}}$ such that for any $n\in\mathbb{N}$ and $x\in\set{0,1}^n$, if $x\in L$ then $\Pr[C_n\br{x}=1]\geq 2/3$, and if $x\notin L$, then $\Pr[C_n\br{x}=0]\geq 2/3$ where $C_n(x)$ is the measurement outcome on the output qubits of the circuit $C_n$ on input $x$. We also introduce the class of \QLCz\ circuits, which consists of \QACz\ circuits with linear-sized ancillae. \QLCz\ is a quantum counterpart of the classical circuit family \LCz\, introduced in~\cite{10.1145/237814.237824}, which is one of the most interesting subclasses of \AC[0] and has received significant attention from various perspectives~\cite{1663737,10.1007/978-3-642-11269-0_6,10.1007/978-3-642-31594-7_65}.




\subsection{State Synthesis}

In this paper, we also investigate the circuit complexity of quantum state synthesis.
In particular, we are interested in the complexity of preparing ``cat states''\footnote{In the literature, the name ``cat states'' appeared in~\cite{548464,DiVincenzo:1996xb,moor}. They are also known as Greenberger–Horne–Zeilinger states (GHZ states)~\cite{Greenberger1989}.} and also ``nekomata'' states introduced in~\cite{rosenthal:LIPIcs.ITCS.2021.32}.
\begin{definition}
    For $n\ge 1$, the cat state, denoted by $\cat{n}$, is defined as
    $$\cat{n} = \f{1}{\sqrt{2}} \p{\ket{0^n} + \ket{1^n}}.$$
\end{definition}
\begin{definition}
    For $n\ge 1$, a quantum state $\ket{\nu}$ is an $n$-nekomata,
    if $\ket{\nu} = \f{1}{\sqrt{2}} \p{\ket{0^n,\psi_0} + \ket{1^n,\psi_1}}$
    for some arbitrary states $\ket{\psi_0}, \ket{\psi_1}$.
\end{definition}

Rosenthal~\cite{rosenthal:LIPIcs.ITCS.2021.32} showed that the syntheses of cat states and $n$-nekomatas are equivalent up to constant depth reductions.
In addition, they are equivalent to computing \parity\ up to constant depth reductions.


\begin{definition}\label{def:state-synthesis}
  We say an $n$-qubit pure state $\psi = \ketbra{\psi}$ is synthesized by a $\QACz$ circuit $U$ with $a$ ancillae and fidelity $1-\delta$,
  if the fidelity between $\psi$ and the first $n$ qubits of $U\ket{0}^{n+a}$ is above $1-\delta$,
  or formally,
  \begin{equation*}
     F\br{\ketbra{\psi}, \Tr_{n+1, \dots, n+a}\Br{U\ketbra{0}^{n+a}U^\dagger}} = \sqrt{\bra{\psi}\Tr_{n+1, \dots, n+a}\Br{U\ketbra{0}^{n+a}U^\dagger}\ket{\psi}} \ge 1 - \delta.
  \end{equation*}
\end{definition}
\begin{definition}
    A family of quantum states $\set{\psi_n}_{n\in\mathbb{N}}$,
    where $\psi_n$ is an $n$-qubit state,
    is in $\stateQACz\Br{\delta}$ if there exists a family of $\QACz$ circuits $\set{U_n}_{n\in\mathbb{N}}$,
    such that for each $n$, $U_n$ synthesizes $\psi_n$ with fidelity $1-\delta$.
    We use $\stateQACz$ to represent $\stateQACz\Br{1/3}$.
    The class $\stateQLCz\Br{\delta}$ and $\stateQLCz$ are defined analogously,
    but with linear ancillae.
\end{definition}

\section{Low Degree Approximation of \texorpdfstring{\QACz}{QAC0} Circuits}
In this section, we present our main technical result.
In \cref{cor:qac0-whole},
we give a sublinear upper bound on the approximate degree of the Heisenberg-evolved measurement operator $UAU^\dagger$,
where $U$ is the unitary implemented by a \QACz\ circuit and $A$ is a low-degree operator.

The main ingredient in our proof is the following low-degree approximation of high-degree tensor product states, adapted from~\cite{KAAV15,Anshu2023concentrationbounds}.
It states that tensor product states such as $\ketbra{0}^{\otimes n}$ can be approximated by a $\sqrt{n}$-local operator.
\begin{lemma}[{\cite[Lemma 3.1]{Anshu2023concentrationbounds}}, see also \cite{KAAV15}]\label{lem:locality-of-eigenstate}
    Let $H=\sum_{i=1}^nH_i$ be a sum of $n$ commuting projectors each acting on $\ell$ qubits,
    and $\ket{\psi}$ be the maximum-energy eigenstate of $H$.
    Then, for any $r\in\br{\sqrt{n}, n}$, let $\varepsilon = 2^{-\frac{r^2}{2^8n}}$,
    $$\degeps{\varepsilon}{\ketbra{\psi}} \le \ell r.$$
\end{lemma}
\begin{cor}\label{lem:locality-of-product}
    Let $\ket{\psi}$ be an $\ell$-qubit pure state.
    Then for any $r\in\br{\sqrt{n}, n}$, let $\varepsilon = 2^{-\frac{r^2}{2^8n}}$. It holds that
    $$\degeps{\varepsilon}{\ketbra{\psi}^{\otimes n}} \le \ell r.$$
\end{cor}
\begin{proof}
    For $i=1, \dots, n$, define $H_i = \id_{i-1}\otimes\ketbra{\psi}\otimes \id_{n-i}$.
    Then clearly $H_i$ are commuting projectors, each acting on $\ell$ qubits.
    In addition, the maximum energy eigenstate of $H=\sum_i H_i$ is precisely $\ket{\psi}^{\otimes n}$.
    So we can apply \cref{lem:locality-of-eigenstate}.
\end{proof}
This approximation also works for \CZGate s, as an $n$-qubit \CZGate\ can be written as
$$\CZG_n = \id_n - 2 \ketbra{1}^{\otimes n}.$$
\begin{cor}\label{cor:approximate-CZGate}
    For any \CZGate\ $\CZG$ acting on $n$ qubits and real number $1 < r < n$,
    there exists an operator $\tilde{\CZG}$ such that
    $$\norm{\CZG - \tilde{\CZG}} \le 2^{1-2^{-8}r}$$
    and
    $$\deg\br{\tilde{\CZG}} \le \sqrt{nr}.$$
\end{cor}
\begin{proof}
  We invoke \cref{lem:locality-of-product} with $\ell \gets 1$ and $r\gets \sqrt{nr}$,
  to get a degree $\sqrt{nr}$ operator $Z$ satisfying $\norm{\ketbra{1}^n - Z} \le 2^{-2^{-8} r}$.
  Then we let $\tilde{\CZG} = \id - 2 Z$.
\end{proof}

The above approximation will be applied to a \QACz\ circuit in a layer-by-layer fashion.

\subsection{Approximating a Single Layer}
The following lemma captures the increase of approximate degrees,
when we apply a layer of \CZGate s to a low-degree operator.
\begin{lemma}\label{lem:qac0-layer}
    Let $n\ge 1$ be an integer
    and $U = \bigotimes_i\CZG_i$ be a layer of \CZGate s acting on totally $n$ qubits.
    For any integer $\ell$ and $r\in (1, n)$, there exists an operator $\tilde{U}$ such that
    $$\norm{U - \tilde{U}} \le \varepsilon = n\cdot2^{1-2^{-8} r}\log e$$
    and for any $2^n\times2^n$ operator $A$ with degree at most $\ell$,
    $$\deg\br{\tilde{U}A\tilde{U}^\dagger} \le 3n^{\frac{2}{3}} \ell^{\frac{1}{3}} r^{\frac{1}{3}}.$$
\end{lemma}

\begin{proof}
    Suppose for each $i$, $\CZG_i$ acts on $s_i$ qubits.
    Let $t = n^{\frac{2}{3}} \ell^{-\frac{2}{3}} r^{\frac{1}{3}}$ be a threshold
    and divide the \CZGate s into the set of small \CZGate s $S=\set{i: s_i\le t}$ and the set of large \CZGate s gates $T=\set{i: s_i > t}$.
    Let $U_S=\bigotimes_{i\in S}\CZG_i$ and $U_T=\bigotimes_{i\in T}\CZG_i$.
    Clearly $U=U_S\otimes U_T$.

    Note that we can assume without loss of generality that $r < t$ because otherwise we have $r\ge t = n^{\frac{2}{3}} \ell^{-\frac{2}{3}} r^{\frac{1}{3}}$,
    implying $n\le \ell r$, and then $\deg\br{\tilde{U}A\tilde{U}^\dagger} \le n \le 3n^{\frac{2}{3}} \ell^{\frac{1}{3}} r^{\frac{1}{3}}$ trivially holds by setting $\tilde{U}=U$.

    Then for any $i\in T$, we have $r < t < s_i$ and we can approximate $\CZG_i$ using \cref{cor:approximate-CZGate} with parameter $r$.
    For each $i\in T$, we get the approximation operator $\tilde{\CZG}_i$ satisfying
    \begin{itemize}
        \item $\norm{\CZG_i - \tilde{\CZG}_i} \le 2^{1-2^{-8r}}$ and
        \item $\deg\br{\tilde{\CZG}_i} \le \sqrt{s_ir}$.
    \end{itemize}
    We define $\tilde{U}_T = \bigotimes_{i\in T}\tilde{\CZG}_i$,
    and our approximation for $U$ will be
    $\tilde{U} = U_S\otimes\tilde{U}_T$.
    We first argue that $U$ and $\tilde{U}$ are close in spectral norm: By \cref{lem:spectral-multiply},
    \begin{equation*}
        \norm{U - \tilde{U}} = \norm{U_S\otimes U_T - U_S\otimes \tilde{U}_T} = \norm{U_T - \tilde{U}_T} = \norm{\bigotimes_{i\in T}\CZG_i - \bigotimes_{i\in T}\tilde{\CZG}_i} \le \br{1 + 2^{1-2^{-8}r}}^{n} - 1 \le \varepsilon.
    \end{equation*}
    Then we upper bound the degree of $\tilde{U}A\tilde{U}^\dagger$.
    By \cref{lem:approximate-local-U-tensor}, we have
    $\deg\br{U_SAU_S^\dagger} \le \ell t$.
    For the large \CZGate s $\tilde{U}_T = \bigotimes_{i\in T}\tilde{\CZG}_i$ part we have
    $$\deg\br{\tilde{U}_T} \le \sum_{i\in T}\deg\br{\tilde{\CZG}_i} \le \sum_{i\in T}\sqrt{s_i r}.$$
    Note that $s_i > t$ for each $i\in T$ and $\sum_is_i\le n$. We have
    $$
        \sum_{i\in T}\sqrt{s_i r} = \sum_{i\in T}s_i\cdot\sqrt{\frac{r}{s_i}} < \sum_{i\in T}s_i\cdot\sqrt{\frac{r}{t}} \le n\sqrt{\frac{r}{t}}.
    $$
    Thus by plugging $t = n^{\frac{2}{3}} \ell^{-\frac{2}{3}} r^{\frac{1}{3}}$,
    $$
        \deg\br{\tilde{U}A\tilde{U}^\dagger} = \deg\br{\tilde{U}_TU_SAU_S^\dagger\tilde{U}_T^\dagger} \le \deg\br{U_SAU_S^\dagger} + 2\deg\br{\tilde{U}_T} \le \ell t + 2 n\sqrt{\frac{r}{t}} \le 3n^{\frac{2}{3}} \ell^{\frac{1}{3}} r^{\frac{1}{3}}.
    $$
\end{proof}

\subsection{Approximating Multiple Layers}
A \QACz\ circuit consists of multiple layers of \CZGate s and single-qubit unitaries.
Note that by \cref{lem:approximate-local-U-tensor},
single-qubit unitaries do not change the degree of an operator.
Hence, only the \CZGate s concern in our case.
Thus, we upper bound the degree of a $\QACz$ circuit by applying \cref{lem:qac0-layer} inductively. 

\begin{lemma}\label{lem:qac0-multiple-layers}
    Let $n\ge 1$ be an integer and $U = L_dM_d\dots L_1M_1L_0$ be a \QACz\ circuit of depth $d$, where $L_i$ consists of only single-qubit unitaries and $M_i$ is the $i$-th layer of \CZGate s.
    For any $r\in(1,n)$ there exists an operator $\tilde{U}$ satisfying
    \begin{equation*}
        \norm{U - \tilde{U}} \le d\cdot n\cdot2^{1-2^{-8} r}\log^2 e
    \end{equation*}
    and for any  $2^n\times 2^n$ operator $A$ with degree at most $\ell$,
    \begin{equation*}
        \deg\br{\tilde{U}A\tilde{U}^\dagger} = O\br{n^{1-3^{-d}}\cdot\ell^{3^{-d}}\cdot r^{1/2}}.
    \end{equation*}

\end{lemma}
For each different layer $i$,
we will carefully choose the parameter $\ell$ for \cref{lem:qac0-layer},
to get an approximation $\tilde{M_i}$ for $M_i$ such that the overall degree will not exceed $\Theta(n)$.
Also, \cref{lem:qac0-layer} ensures that the approximation operator $\tilde{M_i}$ does not depend on
each local monomials of $A$.
Instead of bounding the degree for each monomial individually,
we get a unified unitary $\tilde{M_i}$.
This ensures that the error does not depend on the number of local monomials. Hence, we do not need to worry about the norm of each monomial as concerned in~\cite{Anshu2023concentrationbounds}.
\begin{proof}[Proof of \cref{lem:qac0-multiple-layers}]
    We are only concerned about the $M_i$ gates,
    because by \cref{lem:approximate-local-U-tensor} the single qubit unitaries do not change the degree.
    The idea is to apply \cref{lem:qac0-layer} to each $M_i$ layer by layer, totally $d$ times,
    and get the approximation operators $\tilde{M}_i$ for each $M_i$.
    Then our overall approximation operator will be
    $$\tilde{U} = L_d\tilde{M}_d\dots L_1\tilde{M}_1L_0.$$
    For clarity, for $i=1, 2, \dots, d$, we define
    \begin{equation*}
        U_i = L_iM_iL_{i-1}M_{i-1}\dots L_1M_1L_0,
    \end{equation*}
    \begin{equation*}
        \tilde{U}_i = L_i\tilde{M}_iL_{i-1}\tilde{M}_{i-1}\dots L_1M_1L_0.
    \end{equation*}
    We can check that $U_d = U, \tilde{U}_d = \tilde{U}$ and $U_0 = \tilde{U}_0 = L_0$.

    We can now start our induction argument.
    Initially, we let $\ell_0 = \ell$ and $\varepsilon_0 = 0$.
    For the base case we have
    $$\norm{U_0 - \tilde{U}_0} = \norm{L_0 - L_0} = 0 = \varepsilon_0$$
    and by \cref{lem:approximate-local-U-tensor},
    $$\deg\br{\tilde{U}_0A\tilde{U}_0^\dagger} = \deg\br{L_0AL_0^\dagger} = \deg\br{A} = \ell = \ell_0.$$
    Then, for each $i$'th iteration where $i=1, 2, \dots, d$:
    \begin{itemize}
        \item Before applying \cref{lem:qac0-layer}, we have
        $$\deg\br{\tilde{U}_{i-1}A\tilde{U}_{i-1}^\dagger} = \deg\br{L_{i-1}\tilde{M}_{i-1}\cdots L_0 A L_0^\dagger\cdots \tilde{M}_{i-1}^\dagger L_{i-1}^\dagger} \le \ell_{i-1}$$
        and
        $$\norm{U_{i-1} - \tilde{U}_{i-1}} = \norm{L_{i-1}M_{i-1}\cdots L_0 - L_{i-1}\tilde{M}_{i-1}\cdots L_0} \le \varepsilon_{i-1}.$$
        \item After applying \cref{lem:qac0-layer} with parameter $\ell\gets\ell_{i-1}$, the degree changes from $\ell_{i-1}$ to 
            $$\deg\br{\tilde{U}_iA\tilde{U}_i^\dagger} = \deg\br{L_i\tilde{M}_i\tilde{U}_{i-1}A\tilde{U}_{i-1}^\dagger\tilde{M}_i^\dagger L_i^\dagger} = \deg\br{\tilde{M}_i\tilde{U}_{i-1}A\tilde{U}_{i-1}^\dagger\tilde{M}_i^\dagger} \le \ell_i =  3n^{2/3}\cdot\ell_{i-1}^{1/3}\cdot r^{1/3}.$$
        \item By \cref{lem:spectral-multiply}, the distance changes from $\varepsilon_{i-1}$ to
            $$\norm{U_i - \tilde{U}_i} = \norm{L_{i}M_{i}U_{i-1} - L_{i}\tilde{M}_{i}\tilde{U}_{i-1}} \le \varepsilon_i = \br{1+\varepsilon_{i-1}}\br{1+n\cdot2^{1-2^{-8} r}\log e}-1.$$
    \end{itemize}
    Now we analyze the degree of $\tilde{U}A\tilde{U}^\dagger$:
    By direct calculation,
    $$\ell_i=3^{\frac{3}{2}\br{1-3^{-i}}}n^{1-3^{-i}}\cdot\ell^{3^{-i}}\cdot r^{\frac{1}{2}\br{1-3^{-i}}}.$$
          So the degree of $\tilde{U}A\tilde{U}^\dagger$ is
        $$
          \deg\br{\tilde{U}A\tilde{U}^\dagger} \le \ell_d = 3^{\frac{3}{2}\br{1-3^{-d}}}n^{1-3^{-d}}\cdot\ell^{3^{-d}}\cdot r^{\frac{1}{2}\br{1-3^{-d}}} \le 3^{\frac{3}{2}}n^{1-3^{-d}}\cdot\ell^{3^{-d}}\cdot r^{1/2}.
        $$
    The distance to such a local operator:
        \begin{equation*}
            \varepsilon_i = \br{1 + n\cdot2^{1-2^{-8} r}\log e}^i - 1
        \end{equation*}
        So by using the identity $\br{1 + x}^d \le 1 + xd\log e$ when $x$ is small enough, we have
        \begin{equation*}
            \varepsilon_d = \br{1 + n\cdot2^{1-2^{-8} r}\log e}^d - 1 \le d\cdot n\cdot2^{1-2^{-8} r}\log^2 e.
        \end{equation*}
%
\end{proof}

We can combine \cref{lem:qac0-multiple-layers} and \cref{lem:zeropart} to get the following result
\begin{cor}\label{cor:qac0-whole}
    Let $A$ be a $2^n\times2^n$ operator acting on $n$ qubits with $\norm{A}\le 1$,
    and satisfies $\degeps{\varepsilon}{A} = \ell$.
    Let $U$ be a depth $d$ \QACz circuit working on $a$ qubits.
    Then for $\varepsilon^\prime = \br{1+\varepsilon}\br{1+O(d/n)} - 1$, we have
    $$\degeps{\varepsilon^\prime}{UAU^\dagger} = \tilde{O}\br{n^{1-3^{-d}}\cdot\ell^{3^{-d}}}.$$
    Moreover, let $k\le n$ and $\varphi$ be a $2^k\times2^k$ density operator. It holds that 
    $$\degeps{\varepsilon^\prime}{\Tr_{n-k+1,\dots, n}\Br{UAU^\dagger\br{\id\otimes\varphi}}} = \tilde{O}\br{n^{1-3^{-d}}\cdot\ell^{3^{-d}}}.$$
\end{cor}
\begin{proof}
    Let $\tilde{A}$ be a degree-$\ell$ operator satisfying $\norm{A - \tilde{A}} \le \varepsilon$.
    Let $\tilde{U}$ be the operator obtained from \cref{lem:qac0-multiple-layers} with $r=2^8\br{1 + 2\log n}$.
    Then
    $$\deg\br{\tilde{U}\tilde{A}\tilde{U}^\dagger} \le 3^{3/2}n^{1-3^{-d}}\ell^{3^{-d}} r^{1/2} = \tilde{O}\br{n^{1-3^{-d}}\cdot\ell^{3^{-d}}}.$$
    Moreover, by \cref{lem:spectral-multiply},
    \begin{equation*}
        \norm{UAU^\dagger - \tilde{U}\tilde{A}\tilde{U}^\dagger} \le \br{1+\varepsilon}\br{1 + dn2^{1-2^{-8}r}\log^2 e}^2-1 = \br{1+\varepsilon}\br{1+O(d/n)} - 1.
    \end{equation*}
\end{proof}


\section{Boolean Functions}\label{sec:boolean-functions}
In this section we present our hardness results for computing Boolean functions using \QACz\ circuits.
The starting point is the following theorem,
which states that if a depth-$d$ quantum circuit with $a$ ancillae computes a function $p: \set{0,1}^n\to\mathbb{R}$, then the approximate degree of $p$ is upper bounded by $\tilde{O}((n+a)^{1-{3^{-d}}})$.
This implies that when the number of ancillae is only $n^{1+3^{-d}}$, which is slightly more than linear,
then $(n+a)^{1-3^{-d}} = o(n)$,
and this \QACz\ circuit cannot compute Boolean functions of an approximate degree $\Omega(n)$.
These include the \parity, \maj, and \modk\ functions that we define in \cref{cor:locality-of-parity}.

Let $U$ be the unitary implemented by a $\QACz$ circuit on $n + a$ qubits,
where the first $n$ qubits serve as input,
and the last $a$ input qubits are initialized to the state $\varphi$.
We will upper bound the approximate degree of the operator
$$\Tr_{n+1, \dots, n+a}\Br{U^\dagger\br{\ketbra{1}\otimes\id}U\br{\id\otimes\varphi}},$$
which is actually the projector onto the input space which outputs $1$.
The diagonal entries of this operator are actually the probabilities of outputting $1$ on every classical input,
and corresponds the the matrix $M_p$ defined in \cref{eqn:diagonalf}.

\begin{theorem}\label{thm:qac0-function-degree}
    Let $n\ge 1$ and $U$ be a depth-$d$ \QACz\ circuit
    with $n$ input qubits and $a$ ancillae initialized in state $\varphi$.
    Let $p: \set{0,1}^n\to\mathbb{R}$ be the probability that $U$ outputs $1$.
    That is, for an input $x\in\set{0,1}^n$,
    $$
    p(x) =  \Tr\Br{\br{\ketbra{1}\otimes\id}U\br{\ketbra{x}\otimes\varphi}U^\dagger}.
    $$
    Then $\degeps{\varepsilon}{p} = \tilde{O}\br{(n+a)^{1-3^{-d}}}$ for $\varepsilon=O(d/n)$.
\end{theorem}

Setting $a=O(n)$, we have the following corollary, which extends the analogous result for \LCz\ proved in \cite{10.1145/3444815.3444825}.
\begin{cor}\label{cor:approxdegqlc}
For any $\varepsilon > 0$, it holds that
\[\degeps{\varepsilon}{\QLCz}=o(n).\]
\end{cor}

\begin{proof}[Proof of \cref{thm:qac0-function-degree}]
    Let $A = U^\dagger\br{\ketbra{1}\otimes\id}U$.
    Then we have
    $$p(x) = \Tr\Br{\br{\ketbra{x}\otimes\varphi}A} = \bra{x}\Tr_{n + 1,\dots, n + a}\Br{A\br{\id\otimes\varphi}}\ket{x}.$$
    By \cref{cor:qac0-whole}, we have $\degeps{\varepsilon}{\Tr_{n+1,\dots, n+a}\Br{A\br{\id\otimes\varphi}}} \le \tilde{O}\br{(n+a)^{1-3^{-d}}}$.
    Then the diagonal matrix $M_p$ as defined in~\Cref{eqn:diagonalf} can be obtained by zeroing out all the non-diagonal entries of the matrix
    $$\Tr_{n + 1,\dots, n + a}\Br{A\br{\id\otimes\varphi}}.$$
    So by \cref{lem:diagonaldegree} we have
    \begin{equation*}
        \degeps{\varepsilon}{M_p} \le \degeps{\varepsilon}{\Tr_{n+1,\dots, n+a}\Br{A\br{\id\otimes\varphi}}} \le \tilde{O}\br{(n+a)^{1-3^{-d}}}.
    \end{equation*}
    Finally by \cref{fact:degreecoincide} we have $\degeps{\varepsilon}{p} = \degeps{\varepsilon}{M_p}$ and we prove our theorem.
\end{proof}




With~\cref{thm:qac0-function-degree} we are able to prove that $\QACz$ circuits with linear ancillae
cannot compute high-degree functions.
In the following subsections, we give two flavors of hardness results.
In \cref{sec:worstcase} we prove the worst-case hardness
and in \cref{sec:averagecase} we prove the average-case hardness.
The average-case hardness results follow from the framework of \cite{nadimpalli_pauli_2024}.
Although worst-case hardness results are weaker than average-case hardness results,
they can be applied to a much broader class of Boolean functions.
For example, for the \maj\ function, we prove an optimal $1/2$ worst-case hardness result,
while for the average case,
we can only prove that approximating within a success probability of $1-\frac{1}{\sqrt{n}}$ in the average case is impossible.
For the \modk\ function, we also prove an optimal $1/2$ worst-case hardness result.
However, in the average case, since \modk\ is a biased function,
we can even achieve a success probability of $1-1/k$ simply by always outputting $1$.

The reason for this distinction between worst-case hardness and average-case hardness comes from the fact that
the normalized Frobenius norm can be exponentially smaller than the spectral norm.
So a Boolean function $f$ is much more easily approximated using low-degree functions in the Frobenius norm.
But for the average-case hardness framework from \cite{nadimpalli_pauli_2024} that we are using,
we need a function that is of $\Omega(n)$ degree even if approximated within the Frobenius norm.

In \cref{subsection:boost},
we also show that our worst-case ancillae lower bound of $n^{1+3^{-d}}$
for the $\parity$ function is just one step from a complete resolution of the $\parity\not\in\QACz$ conjecture:
We show that any improvement of our ancillae lower bound to $n^{1+\exp\br{-o(d)}}$,
will imply that $\QACz$ circuits of arbitrary polynomial ancillae can not compute $\parity$.

\subsection{Worst-Case Hardness}\label{sec:worstcase}

With~\cref{thm:qac0-function-degree}, we are able to prove lower bounds on Boolean functions with approximate degrees $\Omega(n)$.

\begin{theorem}\label{thm:main:WorstCase}
    Let $n \ge 1$.
    Suppose that we have a depth $d$ \QACz\ circuit $U$ that computes a function $f: \set{0,1}^n\to\set{0,1}$
    with the worst-case probability $1-\delta$.
    $U$ has $n$ input qubits and $a$ ancillae initialized to the state $\varphi$.
    Then for any $\varepsilon = \delta + O(d/n)$, we have
    \begin{equation*}
          \degeps{\varepsilon}{f} = \tilde{O}\br{(n+a)^{1-3^{-d}}}.
    \end{equation*}
    Moreover, if there exists a \QACz\ circuit that computes a function $f$ satisfying $\degeps{1/3}{f} = \Omega(n)$, with a constant error strictly below $1/2$,
    then we have $a = \tilde{\Omega}\br{n^{1+3^{-d}}}$.
\end{theorem}
\begin{cor}\label{cor:WorstCase}
Suppose a \QACz circuit with depth $d$ and $a$ ancillae computes
a Boolean function $f: \set{0,1}^n\to\set{0,1}$ where $f\in\set{\Parityn, \Majorityn, \Modnk}$
with $2\le k\le cn$ for some $c<1$,
and the worst-case error is a constant below $1/2$, then
$$a = \tilde{\Omega}\br{n^{1+3^{-d}}}.$$
\end{cor}

\begin{proof}[Proof of \cref{thm:main:WorstCase}]
    By \cref{thm:qac0-function-degree}, there exists two Boolean functions $p, \tilde{p}: \set{0,1}\to\mathbb{R}$ satisfying
    \begin{itemize}
        \item $\norm{p - \tilde{p}} \le O(d/n)$.
        \item $\deg\br{\tilde{p}} = \tilde{O}((n+a)^{1-3^{-d}})$.
    \end{itemize}
    Also, since the circuit $U$ computes $f$ with the worst-case error $\delta$, we have $\norm{f - p} \le \delta$.
    This implies $\norm{f - \tilde{p}} \le \norm{f - p} + \norm{p - \tilde{p}} \le \delta + O(d/n)$,
    and this proves the first part of the theorem.

    Given a \QACz\ circuit that computes a Boolean function $f: \set{0, 1}^n\to\mathbb{R}$, with a constant error strictly less than $1/2$ in the worst case, we may reduce the error below $1/3$ using the standard error reduction, i.e., repeating the computation constant times and taking the majority output, which multiplies the depth of the circuit and the size of ancillae by a constant.
    Thus, we still have  $(n+a)^{1-3^{-d}} = \degeps{1/3}{f} = \tilde{\Omega}(n)$, which concludes the proof.
\end{proof}


For $\QACz$ circuits where the depth $d$ is a constant,
the theorem above implies that
if we wish to compute a function with $\Omega(n)$ approximate degree,
and with any nontrivial constant worst-case success probability,
the number of ancilla qubits needs to be superlinear.

\subsection{Average-Case Hardness}\label{sec:averagecase}
In this subsection, we will prove the average-case hardness result for high-degree Boolean functions using arguments similar to those of \cite{nadimpalli_pauli_2024}.
To prove the average-case hardness, we require that the Boolean function be of approximately high degree in terms of Frobenius norm.

The level-$k$ Fourier weight of $f$, denoted by $\Wgt{= k}{f}$, is the sum of the squares of the Fourier coefficients of degree exactly $k$. That is
$$\Wgt{= k}{f} = \sum_{\substack{S\subseteq[n]\\\abs{S}= k}}\widehat{f}(S)^2.$$
The quantites $\Wgt{\le k}{f}, \Wgt{< k}{f}, \Wgt{\ge k}{f}$ and $\Wgt{> k}{f}$ are defined analogously. 
For a Boolean function that has large high-degree Fourier coefficients, that is, $\Wgt{>k}{f}$ is large, it cannot be approximated by a low-degree Boolean function with respect to the Frobenius norm. This also means that its approximate degree is large as the Frobenius norm is an lower bound of the spectral norm.
\begin{lemma}
    Let $f: \set{0,1}^n\to\mathbb{R}$ be a Boolean function.
    If $\Wgt{\ge k}{f} \ge \delta$,
    then for any $\varepsilon < \sqrt{\delta}$, we have $\degeps{\varepsilon}{f} \ge k$.
\end{lemma}
\begin{proof}
    For any function $g: \set{0,1}^n\to\mathbb{R}$ with $\norm{f - g} \le \varepsilon$,
    we have $\normsub{f - g}{2} \le \norm{f - g} \le \varepsilon$.
    Thus, $\Wgt{\ge k}{g} \ge \Wgt{\ge k}{f} - \normsub{f - g}{2}^2 \ge \delta - \varepsilon^2 > 0$.
    This implies $\deg(g) \ge k$.
\end{proof}

By \cref{thm:qac0-function-degree}, we can prove average-case hardness results
for Boolean functions with large high-degree Pauli weights.

\begin{theorem}\label{thm:main:AverageCase}
    Suppose $U$ is a \QACz\ circuit with depth $d$.
    $U$ has $n$ input qubits and $a$ ancillae initialized in state $\varphi$.
    Let $f: \set{0,1}^n\to\set{0,1}$ be any Boolean function.
    It holds that 
    $$\Pr_{\x\sim_U\set{0,1}^n}\Br{C_U(x) = f(x)} \le \frac{1}{2} + \frac{1}{2}\cdot\sqrt{1 - 4\Wgt{> k}{f}} + O(d/n),$$
    for $k \ge \tilde{\Theta}\br{(n+a)^{1-3^{-d}}}$.
    In particular, for $a = \tilde{O}(n^{1+3^{-d}})$,
    $$\Pr_{x\sim_U\set{0,1}^n}\Br{C_U(x) = \Parityn(x)} \le \frac{1}{2} + O(d/n)$$
    and
    $$\Pr_{x\sim_U\set{0,1}^n}\Br{C_U(x) = \Majorityn(x)} \le 1 - \Omega\br{\frac{1}{\sqrt{n}}} + O(d/n).$$
\end{theorem}
\begin{remark}
    This theorem also holds in the case where we want to compute a Boolean function with randomized output.
    In this case, the function $f$ now has the form $f: \set{0,1}^n\to\mathbb{R}$,
    where for each input $x\in\set{0,1}^n$,
    the value $f(x)$ denotes the probability of output $1$.
    The proof is literally the same.
\end{remark}

\begin{proof}[Proof of \cref{thm:main:AverageCase}]
  For any $x\in\set{0,1}^n$, let $p\br{x}$ be the probability that the circuit outputs $1$ on input $x$.
  By \cref{thm:qac0-function-degree}, 
  $\degeps{\varepsilon}{p} = \tilde{O}\br{(n+a)^{1-3^{-d}}}$.
  Then there exists a function $\tilde{p}: \set{0,1}^n\to\mathbb{R}$ satisfying
  $\norm{p - \tilde{p}} \le \varepsilon$ and $\deg\br{\tilde{p}} \le \tilde{O}\br{(n+a)^{1-3^{-d}}}\le k$.
  Let $q = 2p - 1, \tilde{q} = 2\tilde{p}-1$ and $g = 2f - 1$.
  Then $\normsub{q - \tilde{q}}{2} \le \norm{q - \tilde{q}} \le 2\varepsilon$.
  Let $x$ be a random variable uniformly distributed over $\set{0,1}^n$.
  \begin{align*}
    2\cdot\Pr_{x\sim_U\set{0,1}^n}\Br{C_U(x) = f(x)} - 1 &= \expec{x\sim_U\set{0,1}^n}{q(x)g(x)} \\
        &= \sum_{S\subseteq[n]}\widehat{q}(S)\widehat{g}(S) \\
        &= \sum_{\abs{S}\le k}\widehat{q}(S)\widehat{g}(S) + \sum_{\abs{S} > k}\widehat{q}(S)\widehat{g}(S) \\
        &\le \sqrt{\Wgt{\le k}{q}\Wgt{\le k}{g}} + \sqrt{\Wgt{> k}{q}\Wgt{> k}{g}} \\
        &\le \sqrt{\Wgt{\le k}{g}} + \normsub{q - \tilde{q}}{2}\sqrt{\Wgt{> k}{g}} \\
        &\le \sqrt{\Wgt{\le k}{g}} + 2\varepsilon\sqrt{\Wgt{> k}{g}}.
  \end{align*}
  The second inequality follows since $\tilde{q}$ has degree below $k$,
  so $\sqrt{\Wgt{>k}{q}} = \sqrt{\Wgt{>k}{q - \tilde{q}}} \le \normsub{q - \tilde{q}}{2}$.
  Tidying up, we have
  \begin{align*}
    \Pr_{x\sim_U\set{0,1}^n}\Br{C_U(x) = f(x)} &\le \frac{1}{2} + \frac{1}{2}\cdot\sqrt{\Wgt{\le k}{g}} + \varepsilon\cdot\sqrt{\Wgt{>k}{g}} \\
        &\le \frac{1}{2} + \frac{1}{2}\cdot\sqrt{1 - \Wgt{> k}{g}} + \varepsilon\cdot\sqrt{\Wgt{>k}{g}} \\
        &\le \frac{1}{2} + \frac{1}{2}\cdot\sqrt{1 - 4\Wgt{> k}{f}} + \varepsilon. \\
  \end{align*}
  Notice that $\Parityn = \frac{1}{2} + \frac{1}{2}\chi_{[n]}$.
  So we have $\Wgt{> k}{\Parityn} = \frac{1}{4}$ for $k<n$.
  For the Majority function, we have $\Wgt{> k}{\Majorityn} \ge \Omega\br{\frac{1}{\sqrt{k}}}$ for odd $k$~\cite[Equation 5.11]{ODonnell2014}.
\end{proof}




This theorem is an analog of Proposition 32 and Theorem 33 of \cite{nadimpalli_pauli_2024}.
One main difference is that we use \cref{thm:qac0-function-degree} to directly upper bound
the degree of the Boolean function,
instead of considering the Choi representation of quantum channels.
This turns out to be more straightforward and efficient in our applications.

The result of \cite{nadimpalli_pauli_2024} only allows $n^{\frac{1}{d}}$-qubit ancillary pure state, the barrier being the increase of the Frobenius norm when applying the operation
$P\mapsto\Tr_{n+1,\dots,n+a}\Br{P\br{\id\otimes\varphi}}$.
We overcome this barrier by using spectral norm approximations.
However, our method creates new barriers.
Since we are not ``removing'' the large \CZGate s as in \cite{nadimpalli_pauli_2024},
in the case where there are too many, say $\omega(n)$ large \CZGate s,
our method immediately fails.
Also, it is still open whether we can approximate $\AC[0]$ circuits with degree $o(n)$ polynomials~\cite{10.1145/3444815.3444825}.
It would be even more challenging to do so with quantum circuits.

Rosenthal~\cite{rosenthal:LIPIcs.ITCS.2021.32} proved that computing \Parityn~ is equivalent to synthesizing the $n$-qubit cat state.
It is interesting to note that our hardness results \cref{thm:main:WorstCase} and \cref{thm:main:AverageCase} hold for any linear-sized ancillae, which implies that only providing a cat state in the ancillae is not enough to compute \parity.
Instead, a quantum circuit that synthesizes the cat state is needed for the reduction.


\subsection{Application: Agnostic Learning for \texorpdfstring{\QLCz}{QLC0}}
Kearns, Schapire and Sllie~\cite{10.1145/130385.130424} proposed the agnostic PAC learning model, which is a more general model than the standard PAC learning model. Let $\mathcal{D}$ be a distribution on $\set{0,1}^n\times\set{0,1}$ and let $\mathcal{C}$ be a concept class. For any $h:\set{0,1}^n\rightarrow\set{0,1}$, the error is defined to be $\mathrm{err}_{\mathcal{D}}\br{h}=\mathrm{Pr}_{\br{x,y}\sim\mathcal{D}}\Br{h\br{x}\neq y}$ and one defines $\mathrm{opt}=\min_{c\in\mathcal{C}}\mathrm{err}_{\mathcal{D}}\br{c}$. We say that $\mathcal{C}$ is agnostically learnable in time $T\br{n,\varepsilon,\delta}$, if there exists an algorithm that takes as input $n,\delta$ and has access to an example oracle of $\mathcal{D}$, and satisfies the following properties: It runs in time at most $T\br{n,\varepsilon,\delta}$ and with probability at least $1-\delta$, it outputs a hypothesis $h$ that satisfies $\mathrm{err}_{\mathcal{D}}\br{h}\leq\mathrm{opt}+\varepsilon$.  The agnostic model is believed to be closer to the realistic scenario than the standard PAC model. However, designing efficient agnostic learning algorithms is challenging in general. Even very few concept classes are known to be agnostically learnable in subexponential time.  
Bun, Kothari, and Thaler~\cite{doi:10.1137/1.9781611975482.42} gave a subexponential-time agnostic learning algorithm for the class of functions with approximate degree $n^c$ for $c<1$, built on~\cite{doi:10.1137/060649057}.
\begin{fact}[{\cite[Corollary 24]{doi:10.1137/1.9781611975482.42}}]
Let $C$ be a set of Boolean functions on $\set{0,1}^n$.
Suppose that for every $c\in C$, the $\varepsilon$-approximate degree of $c$ is at most $d$.
Then for every $\delta > 0$, there is an algorithm running in time $\poly\br{n^d, 1/\varepsilon, \log(1/\delta)}$
that agnostically learns $C$ to error $\varepsilon$ with respect to any (unknown) distribution D over $\set{0,1}^n\times\set{0,1}$.
\end{fact}
As a simple corollary of \cref{thm:qac0-function-degree},
we know that the classical functions that a depth-$d$ $\QLCz$ circuit computes has approximate degree $\tilde{O}(n^{1-3^{-d}})$. As an immediate corollary, we have a subexponential learning algorithm for the functions in \QLCz\ in the distribution-free agnostic PAC learning model, generalizing the classical counterpart~\cite[Theorem 7]{doi:10.1137/1.9781611975482.42}.
\begin{theorem}
    The concept class of $n$-bit functions computed by $\QLCz$ circuits of depth $d$
    can be learned in the distribution-free agnostic PAC model in time $2^{\tilde{O}\br{n^{1-3^{-d}}}}$.
\end{theorem}

\subsection{\texorpdfstring{$\QACz$}{QAC0} with Low-Degree Classical Post-processing}

For a quantum channel, we measure the output qubits in the computational basis.
\cref{thm:qac0-function-degree} relies on the fact that we only look at the first qubit of the output of a quantum circuit.
That is, we consider the projector $\ketbra{0}\otimes\id$, which is of degree $1$.
We can allow more complicated classical post-processing of the quantum output,
e.g. linear-size $\AC[0]$ circuits,
as long as the post-processing is low degree.
By slightly modifying the proof of \cref{thm:qac0-function-degree},
we obtain the following theorem.
\begin{theorem}\label{thm:ac0-circ-qac0}
    Let $n\ge 1$.
    Let $U$ be a depth-$d$ \QACz\ circuit with $n$ input qubits and $a$ ancillae initialized in the state $\varphi$.
    Let $f: \set{0,1}^{n+a}\to\mathbb{R}$ be a Boolean function satisfying $\degeps{\delta}{f} = \ell$.
    Let $p: \set{0,1}^n\to\mathbb{R}$ be the function describing the probability that $f\circ U$ outputs $1$.
    That is, for an input $x\in\set{0,1}^n$,
    $$
    p(x) =  \Tr\Br{M_f\cdot U\br{\ketbra{x}\otimes\varphi}U^\dagger}.
    $$
    Then $\degeps{\varepsilon}{p} = \tilde{O}\br{(n+a)^{1-3^{-d}}\cdot\ell^{3^{-d}}}$ for $\varepsilon=\br{1 + \delta}\br{1 + O(d/n)} - 1$.
    In particular, it implies that $L\notin\LCz\circ\QLCz$ for $L\in\set{\mathrm{PARITY},\mathrm{MAJORITY},\mathrm{MOD}_k}$,
   where $0<c<1$ and $2\leq k\leq cn$.
\end{theorem}
The proof is almost identical to the proof of \cref{thm:qac0-function-degree},
we provide a proof for completeness in \cref{appendix:proofs}.
This improves upon the results of Slote~\cite{slote:LIPIcs.ITCS.2024.92},
where he proved a similar result for $\mathrm{QNC}^0$ pre-processing.

\section{Towards \texorpdfstring{$\parity\not\in\QACz$}{Parity Not In QAC0}}\label{subsection:boost}
In this section, we show that our result \cref{cor:WorstCase} is just one step away from the conjecture that $\parity\notin\QACz$, despite that our ``barely superlinear'' lower bound of $\tilde{\Omega}\br{n^{1+3^{-d}}}$ seems far away from the arbitrary polynomial ancillae allowed in $\QACz$ circuits.
In particular, we show that $n^{1+\exp\br{-\Theta(d)}}$ ancillae is a threshold.
Any ancillae lower bound of the form $n^{1+\exp\br{-o(d)}}$,
e.g., $n^{1+\exp\br{d/\log d}}$,
will lead to a lower bound for any polynomial ancillae.
To see this, we show that given a $\QACz$ circuit family with depth $d$
and ancillae $n^c$,
we can get a $\QACz$ circuit computing $\Parityn$ with depth $O(cdt)$ and ancillae $n^{1+2^{-t}}$,
for any constant $t$.

\begin{definition}
    A $\circsize{d}{n}{a}$ circuit is a $\QACz$ circuit with depth $d$,
    input size $n$ and ancillae size $a$.
\end{definition}

We first present the main theorem of this section,
which we will prove by directly combining \cref{lem:boost-step-1} and \cref{lem:boost-step-2} at the end of this section.
\begin{restatable}{theorem}{restateboost}\label{thm:boost-full}
    Suppose there exist constants $d\in\mathbb{Z}^+, c \in\mathbb{Z}^+, N_0\in\mathbb{Z}^+$,
    and a $\QACz$ circuit family $\set{U_n}_{n\ge N_0}$ such that for each $n \ge N_0$,
    the circuit $U_n$ is a $\circsize{d}{n}{n^c}$ circuit,
    and computes $\Parityn$ with the worst-case error $\negl(n)$.
    Then there exists $D=O(cd)$, for any $K\ge 1$ and infinitely many $n$,
    we can construct a $\circsize{KD}{n}{n^{1+2^{-K}}}$ circuit,
    which computes \Parityn with the worst-case error $\negl(n)$.
\end{restatable}

\begin{cor}\label{cor:arbitrary}
    For any function $\delta: \mathbb{R}\to\mathbb{R}$ satisfying $\lim_{x\to\infty}\delta(x)=\infty$,
    if any $\QACz$ circuit with $n^{1+\exp\br{-d/\delta(d)}}$ ancillae,
    where $d$ is the depth of this circuit family,
    can not compute $\Parityn$ with the worst-case error $\negl(n)$,
    then any $\QACz$ circuit family, where arbitrary polynomial ancillae is allowed,
    can not compute $\Parityn$ with the worst-case error $\negl(n)$.
\end{cor}
\begin{proof}
    Suppose on the contrary, there exists a $\QACz$ circuit family that computes $\Parityn$
    with the worst-case error $\negl(n)$.
    Then let $D$ be the constant in \cref{thm:boost-full},
    there exists a large enough constant $K$ such that $D/\delta(KD) \le 1/2$, and then
    \begin{equation*}
        n^{1+\exp\br{-KD/\delta(KD)}} \ge n^{1+\exp\br{-K/2}} \ge n^{1+2^{-K}}.
    \end{equation*}
    By \cref{thm:boost-full},
    we can construct a $\QACz$ circuit with depth $KD$ and $n^{1+2^{-K}} \le n^{1+\exp\br{-(KD)/\delta(KD)}}$ ancillae, a contradiction.
\end{proof}

\subsection{Proof Overview of Theorem \ref{thm:boost-full}}
For any $c > 1$, given a $\circsize{d}{n}{n^c}$ circuit $U_n$,
we wish to construct a $\circsize{KD}{N}{N^{1+2^{-K}}}$ circuit for some constant $D$, any $K$, and infinitely many $N$.
We use the idea that $\Parity{m}$ can be computed recursively. Thus, we can construct a circuit to compute $\Parity{m}$ for some $m > n$
using $U_n$ as a building block. In this way of computing parity, ancillae size grows slower than $m^c$.
We first give a very simple example.
For $m = 2n$,
we can construct a $2$-layer circuit that computes $\Parity{m}$ with depth $d+1$ and $2n^c$ ancillae,
as shown in \cref{fig:example-2n}:
\begin{figure}[h]
    \begin{equation*}
       \Qcircuit @C=1em @R=.7em {
       \lstick{\ket{x_1}} & \multigate{2}{U_n} & \qw & \targ & \meter \\
       \lstick{\vdots} & \nghost{U_n} & \vdots &  & \\
       \lstick{\ket{x_n}} & \ghost{U_n} & \qw & \qw & \qw \\
       \lstick{\ket{x_{n+1}}} & \multigate{2}{U_n} & \qw & \ctrl{-3} & \qw \\
       \lstick{\vdots} & \nghost{U_n} & \vdots & & \\
       \lstick{\ket{x_{2n}}} & \ghost{U_n} & \qw & \qw & \qw \\
    }
    \end{equation*}
    \Description{A quantum circuit composing of two identical sub-circuits, a CNOT gate and a final measurement on the first qubit.}
    \caption{Circuit computing $\Parity{m}$ for $m=2n$}
    \label{fig:example-2n}
\end{figure}
In the first layer, we use one instance of $U_n$ to compute the parity of the first half of the input bits,
and use another instance of $U_n$ to compute the parity of the second half of the input bits.
In the second layer, we get the parity of the whole input by computing the XOR of these two outputs.

The advantage of this construction is that now the ancillae size is $2n^c$,
which is smaller than $m^c = 2^cn^c$ since $c > 1$.
That is, the ancillae size is decreased relative to the input size.
We can now amplify this advantage by constructing larger circuits.
In the second layer, we can compute the parity of $n$ bits instead of $2$ bits,
by replacing the CNOT gate into an instance of $U_n$.
Then we can put $n$ instances of $U_n$ instead of $2$ instances in the first layer.
This yields a circuit computing $\Parity{m}$ for $m = n^2$,
with depth $2d$ and $(n+1)\cdot n^c \approx n^{c+1}$ ancillae.
Moreover, by repeating this construction and
further considering $k$-layer circuits as in \cref{fig:boost-step-1},
where the bottom layer has $n^{k-1}$ instances of $U_n$,
we con construct circuits that compute $\Parity{m}$ for $m=n^k$,
that have depth $kd$, and use approximately $n^{c+k-1}$ ancillae.
Note that
\begin{equation*}
    n^{c+k-1} = m^{\frac{c+k-1}{k}} = m^{1 + \frac{c-1}{k}},
\end{equation*}
so now for any $k$, we have a $\circsize{kd}{m}{m^{1+\frac{c-1}{k}}}$ circuit computing $\Parity{m}$.
For any constant $\varepsilon > 0$ that is independent of the depth of the circuit,
we can get a $\circsize{O(1)}{m}{m^{1+\varepsilon}}$ circuit by choosing $k = (c-1)/\varepsilon$.

In the $k$-layer construction above, the depth of the circuit grows linearly with respect to $k$. The best we could expect with this construction is a $\circsize{d}{n}{n^{1+\varepsilon}}$ circuit for $\varepsilon=d^{-c}$ ($0<c<1$). 
For instance, if $\varepsilon = d^{-0.99}$, then 
we can choose a constant $k$ satisfying
\begin{equation*}
    1 + \frac{c-1}{k} \le 1 + \varepsilon = 1 + (kd)^{-0.99}.
\end{equation*}

To further compress the size of ancillae, we adapt a more efficient construction.
In particular, since the construction above alone can not prove \cref{thm:boost-full},
for the sake of clean parameters,
we will choose $k=c$, and get a $\circsize{O(1)}{m}{m^2}$ circuit.

Now the goal becomes reducing the ancillae size from $n^2$ to $n^{1+\exp\br{-o(d)}}$.
Looking back to the $k$-layer construction in \cref{fig:boost-step-1},
we observe that we are using the ancillae very inefficiently when the number of layers grows:
In the first (bottom) layer, we use $n^{k-1}$ instances of $U_n$ and thus there are $n^{c+k-1}$ ancillae.
However, in the second layer, we use $n^{k-2}$ instances of $U_n$ and thus the ancillae used is only
$n^{c+k-2} = \frac{1}{n}\cdot n^{c+k-1}$.
The top layer is the worst, which uses only $n^c$ ancillae. This indicates that
an improvement may be possible if we can adjust the ancillae size in each layer to make them equal.
This would not change the order of magnitude of the total number of ancillae.
However, we might be able to compute the parity function for larger input sizes,
thus effectively reducing the ancilla size with respect to the input size.

This improvement can be easily implemented when we have a $\QACz$ circuit family $\set{U_n}$,
where each $U_n$ is a $\circsize{d}{n}{n^2}$ circuit.
We also know that such a circuit family is promised by the construction in \cref{fig:boost-step-1}.
We illustrate the improved $k$-layer construction in \cref{fig:boost-step-2}.
In the first (bottom) layer, we fill it with instances of the circuit $U_n$.
In the second layer, we fill it with instances of the circuit $U_{n^2}$.
In the $i$th layer, we fill it with instances of the circuit $U_{n^{2^{i-1}}}$.
Now this construction computes $\Parity{m}$ for $$m=\prod_{i=1}^{k}n^{2^{i-1}} = n^{2^k-1}.$$
And the number of instances of $U_{n^{2^{i-1}}}$ in the $i$th layer can be calculated to be
$$\prod_{j=i+1}^kn^{2^{j-1}} = n^{2^k - 2^{i}}.$$
So the ancillae size in the $i$th layer is
$$n^{2^k - 2^{i}} \cdot \br{n^{2^{i-1}}}^2 = n^{2^k - 2^{i}} \cdot n^{2^{i}} = n^{2^k}.$$
And the total ancillae size is $k\cdot n^{2^{k}} \approx n^{2^{k}}$.
The depth of this construction is $kd$.
In Conclusion, we get a $\circsize{kd}{n^{2^k-1}}{n^{2^k}}$ circuit computing $\Parity{2^{2^k-1}}$.
With $m=n^{2^{k}-1}$, we have
\begin{equation*}
    n^{2^k} = m^{\frac{2^k}{2^k-1}} = m^{1 + \frac{1}{2^k-1}} \approx m^{1 + 2^{-k}}.
\end{equation*}
So we have a $\circsize{kd}{m}{m^{1+2^{-k}}}$ circuit computing $\Parity{m}$, proving \cref{thm:boost-full}.

There are still two caveats.
The first one is regarded to the error.
The above construction seems to be correct only when each $U_n$ computes $\Parity{n}$ without any error.
When each $U_n$ computes $\Parity{n}$ with worst case probability $\br{1+\delta}/2$,
and there are totally $t$ instances of these circuits,
we show that the total correct probability is $\br{1+\delta^t}/2$.
Since $t$ will always be a polynomial in $n$,
the above construction is stable as long as $\delta = 1-\negl(n)$.
Second, for this accept probability to hold,
we need to measure the first output qubit of each $U_n$ before feeding them to the next layer.
Although measurements are in general not allowed in $\QACz$ circuits,
we can achieve the same effect by applying a CNOT gate controlled on the first output qubit of each $U_n$,
targeted to an ancilla initialized to the state $\ket{0}$.
These coherent operations slightly increase the overall depth and acillae size of our constructions.
See \cref{lem:qac0-boost} for details.

Therefore, the overall proof of \cref{thm:boost-full} consists of two steps. The first step is to reduce an arbitrary \QACz circuit computing \parity~ to an $\circsize{O(1)}{n}{n^2}$ circuit. The second step further reduces the circuit to $\circsize{d}{n}{n^{1+\exp(-o(d))}}$. It is worth noticing that if we optimize the usage of ancillae in the first step, as we did in the second step, we could also get $\circsize{d}{n}{n^{1+\exp(-o(d))}}$. However, the parameters would be complicated. To keep the clearness of the parameters, we maintain this two-step reduction in our proof.

\subsection{Proof of Theorem~\ref{thm:boost-full}}

The rest of this section is devoted to proving \cref{thm:boost-full}.
The building block of our construction is the following lemma,
which states that given two $\QACz$ circuits computing $\Parity{n_1}$
and $\Parity{n_2}$ respectively,
we can construct a $\QACz$ circuit computing $\Parity{n_1n_2}$.
\begin{lemma}\label{lem:qac0-boost}
    Suppose there exist two $\QACz$ circuits $U_t$ and $U_b$\footnote{Here `t' stands for top and `b' stands for bottom.}.
    such that the circuit $U_t$ (resp. $U_b$) is a $\circsize{d_t}{n_t}{a_t}$ (resp. $\circsize{d_b}{n_b}{a_b}$) circuit,
    and computes $\Parity{n_t}$ (resp. $\Parity{n_b}$) with the worst-case probability at least $(1 + \delta_t)/2$ (resp. $\br{1+\delta_b}/2$).
    Then we can construct a $\QACz$ circuit $U$,
    such that $U$ is a $\circsize{d_t + d_b +1}{n_tn_b}{n_t\br{a_b+1}+a_t}$ circuit,
    and computes $\Parity{n_tn_b}$ with the worst-case probability at least $(1 + \delta_b^{n_t}\delta_t)/2$.
\end{lemma}
\begin{remark}
    Note that the term $\delta_b^{n_t}$ in the worst-case probability
    may decrease exponentially fast if $\delta_b$ is smaller than $1 - 1/n_t$,
    e.g., when $\delta_b$ is a constant.
    Also, the input size $n_t$ can be as large as any polynomial in the constructions below.
    So to keep the worst-case probability of correctly computing $\Parity{n_tn_b}$ non-trivial,
    we will always assume $\delta_b = 1 - \negl(n)$ for some negligible function.
\end{remark}
\begin{proof}[Proof of \cref{lem:qac0-boost}]
    Let $n = n_tn_b$.
    Observe that for $x\in\set{0,1}^{n}$,
    we can write it as $x=x_1\dots x_{n_t}$ where each $x_i\in\set{0,1}^{n_b}$,
    and we can implement the function $\Parity{n}$ by $n_t$ instances of $\Parity{n_b}$,
    followed by one instance of $\Parity{n_t}$:
    \begin{equation*}
        \Parity{n}(x) = \Parity{n_t}\br{\Parity{n_b}\br{x_1}, \dots, \Parity{n_b}\br{x_{n_t}}}.
    \end{equation*}
    Also, the evaluations of $\Parity{n_b}\br{x_1}, \dots, \Parity{n_b}\br{x_n}$ are independent to each other.
    So they can be carried out in parallel.
    This suggests a three-layer circuit for $\Parity{n}$,
    by using the circuits $U_t$ and $U_b$ as subroutines:
    In the first layer,
    we compute each $\Parity{n_b}\br{x_i}$ in parallel using $n_t$ instances of the circuit $U_b$.
    In the second layer, we need to measure each output of $U_b$ in the computational basis.
    This can be achieved by applying a CNOT gate,
    controlled on the qubit to be measured,
    to an ancilla initialized to the state $\ket{0}$.
    Finally in the third layer, we compute the parity of these outputs,
    by an instance of the circuit $U_t$.
   
    The circuit of this construction for $n_t=5, n_b=4$ can then be depicted in \cref{fig:paritynsquare}.
    \begin{figure}
        \centering
        \subfloat[qubit layout]{
            \begin{tikzpicture}
                \draw[thick, step=0.5] (-0.5, 0) grid (2, 2.5);
                \foreach \i in {1,...,5}
                {
                    \node at (-0.25, \i * 0.5 - 0.25) {\scriptsize $\ket{0}$};
                    \foreach \j in {1,...,4}
                    {
                        \node at (0.5 * \j - 0.25, 2.75 - 0.5 * \i) {\scriptsize $x_{\i, \j}$};
                    }
                }
            \end{tikzpicture}
            \label{subfig:paritynsquare:layout}
        }

        \subfloat[first layer]{
            \begin{tikzpicture}
                \draw[step=0.5, dotted] (0, 0) grid (2, 2.5);
                \draw[thick, step=0.5] (-0.5, 0) grid (0, 2.5);
                \draw[thick, xstep=2, ystep=0.5] (0, 0) grid (2, 2.5);
                \foreach \i in {1,...,5}
                {
                    \node at (0.75, \i * 0.5 - 0.25) {\scriptsize $U_b$};
                }
            \end{tikzpicture}
            \label{subfig:paritynsquare:firstlayer}
        }
        \subfloat[second layer]{
            \begin{tikzpicture}
                \draw[step=0.5, dotted] (0, 0) grid (1, 2.5);
                \draw[thick, xstep=1, ystep=0.5] (0, 0) grid (1, 2.5);
                \draw[thick, step=0.5] (1, 0) grid (2.5, 2.5);
                \foreach \i in {1,...,5}
                {
                    \node at (0.25, \i * 0.5 - 0.25) {$\oplus$};
                    \filldraw[black] (0.75, \i * 0.5 - 0.25) circle (1pt);
                    \draw[-] (0.25, \i * 0.5 - 0.25) -- (0.75, \i * 0.5 - 0.25);
                }
            \end{tikzpicture}
            \label{subfig:paritynsquare:secondlayer}
        }
        \subfloat[third layer]{
            \begin{tikzpicture}
                \draw[step=0.5, dotted] (0, 0) grid (2, 2.5);
                \draw[thick, step=0.5] (-0.5, 0) grid (0, 2.5);
                \draw[thick, xstep=0.5, ystep=2.5] (0, 0) grid (0.5, 2.5);
                \draw[thick, step=0.5] (0.5, 0) grid (2, 2.5);
                \node at (0.25, 1.25) {\scriptsize $U_t$};
            \end{tikzpicture}
            \label{subfig:paritynsquare:thirdlayer}
        }
        \Description{A group of four grids. The first grid describes the qubit layout of the circuit. The other grids describe different layers of quantum gates in the circuit.}
        \caption{Circuit for $\Parity{n_1n_2}$}
        \label{fig:paritynsquare}
    \end{figure}
    The ancillae used by the circuits $U_t$ and $U_b$ are omitted for clarity.
    We partition the $n$ inputs qubits into an $n_t\times n_b$ grid as in \cref{subfig:paritynsquare:layout}.
    In the first layer (\cref{subfig:paritynsquare:firstlayer}),
    we apply the circuit $U_b$ to each row,
    and the output qubit is the left-most qubit.
    In the second layer (\cref{subfig:paritynsquare:secondlayer}),
    we ``measure'' the output of each row.
    Each ``measurement'' is implemented by a CNOT gate plus an ancilla initialized in the state $\ket{0}$.
    In the third and final layer (\cref{subfig:paritynsquare:thirdlayer}),
    we apply an instance of the circuit $U_t$,
    on the input qubits in the left-most column.
    Now the result of $\Parity{n}(x)$ should be stored in the top-left qubit,
    originally containing $x_{1,1}$, the first bit of the input.
    To simplify calculation,
    we assume that the circuits $U_i$ for $i\in\set{b, t}$ computes $\Parity{n_i}$ for any input with probability exactly $\br{1+\delta_i}/2$.
    This is a reasonable assumption because otherwise,
    the circuit $U_i$ does better for some inputs, and the overall correct probability with be even larger.
    Now fix any input.
    Let $\ket{\varphi_i} = \alpha_{i, 0}\ket{0}\ket{\mu_{i}} + \alpha_{i, 1}\ket{1}\ket{\nu_{i}}$
    be the output of the circuit computing $\Parity{n_b}(x_i)$.
    Let $b_i = \Parity{n_b}\br{x_i}$.
    Then $\abs{\alpha_{i, b_i}}^2 = \br{1+\delta}/2$.
    After the second layer consisting of CNOT gates, the output state becomes
    $\abs{\alpha_{i, 0}}^2\ketbra{0} + \abs{\alpha_{i, 1}}^2\ketbra{1}$.

    For $i\in\set{t,b}$,
    let $\pi_i$ be a distribution on $\set{-1,1}$ such that
    $\pi_i(1) = \br{1+\delta_i}/2$ and $\pi_i(-1) = \br{1-\delta_i}/2$.
    For $i=1, \dots, n$, let $w_i$ be independent random variables,
    such that $w_i = 1$ if and only if the measurement outcome of the circuit computing $\Parity{n_b}(x_i)$ is correct.
    Let $v$ be an independent random variable such that $v = 1$ if and only if the output of the $U_t$ in the third layer is correct.
    Then $w_i$ are all distributed according to the distribution $\pi_b$,
    and $v$ is distributed according to the distribution $\pi_t$.
    Also, the circuit $U$ computes $\Parity{n}(x)$ correctly if and only if
    $$v\cdot\prod_{i=1}^{n_t}w_i = 1.$$

    Then the circuit $U$ correctly computes $\Parity{n}$ with probability exactly
    \begin{align*}
        \Pr\Br{v\cdot\prod_{i=1}^{n_t}w_i = 1}
            &= \frac{\Pr\Br{v\cdot\prod_iw_i = 1} + \br{1 - \Pr\Br{v\cdot\prod_iw_i = -1}}}{2} \\
            &= \frac{1 + \expec{}{v\cdot\prod_iw_i}}{2} \\
            &= \frac{1 + \expec{}{v}\cdot\prod_i\expec{}{w_i}}{2} \\
            &= \frac{1 + \delta_b^{n_t}\delta_t}{2}.
    \end{align*}
\end{proof}



The proof of \cref{thm:boost-full} will be split in two steps:
\begin{itemize}
    \item In step 1, we show how to construct a $\QACz$ circuit with $n^{2}$ ancillae,
            from any $\QACz$ circuit with arbitrary polynomial ancillae.
    \item In step 2, given a circuit family using at most $n^2$ ancillae,
    we show how to construct $\QACz$ circuits with $n^{1+\exp\br{-o(d)}}$ ancillae.
\end{itemize}

\paragraph{Step 1}
Given a $\QACz$ circuit $U_n$ computing $\Parityn$ with $n^c$ ancillae,
the goal of the first step is to construct a $\QACz$ circuit computing $\Parity{n^c}$ with $n^{2c}$ ancillae,
using the construction in \cref{fig:boost-step-1}:
\begin{figure}
    \centering
    \begin{tikzpicture}
        \node[circnode] (T0) at (10,5) {$U_n$};
        \node[circnode] (T11) at (7.5,4) {$U_n$};
        \node[circnode] (T12) at (12.5,4) {$U_n$};
        \node[circnode] (T21) at (5,3) {$U_n$};
        \node[circnode] (T22) at (9,3) {$U_n$};
        \node[circnode] (T23) at (15,3) {$U_n$};
        \node[circnode] (Tn1) at (4,1) {$U_{n}$};
        \node[circnode] (Tn1) at (10,1) {$U_{n}$};
        \node[circnode] (Tn1) at (16,1) {$U_{n}$};
        
        \draw[->] (T11) -- (T0);
        \draw[->] (T12) -- (T0);
        \draw[->] (T21) -- (T11);
        \draw[->] (T22) -- (T11);
        \draw[->] (T23) -- (T12);

        \node[] at (10,4) {$\cdots$};
        \node[] at (7,3) {$\cdots$};
        \node[] at (12,3) {$\cdots$};
        \node[] at (6,2) {$\vdots$};
        \node[] at (10,2) {$\vdots$};
        \node[] at (14,2) {$\vdots$};
        \path[draw,decorate,decoration=brace] (2.5,1) -- (2.5,6) node[midway,left]{$c$ layers};
    \end{tikzpicture}
    \Description{A tree composing of unitaries}
    \caption{Step 1, $n^c\to n^2$ ancillae}
    \label{fig:boost-step-1}
\end{figure}
We consider a $c$-layer $\QACz$ circuit computing $\Parity{n^c}$.
In the top layer, there is a single circuit $U_n$ that computes $\Parityn$.
Following, from top to down,
for $i=2, \dots, c$, each layer consists of $n^{i-1}$ instances of the circuit $U_n$.
These layers are connected by \cref{lem:qac0-boost}.
So the overall depth will be $cd+c-1$.
The total number of instances of $U_n$ is
\begin{equation*}
    1 + n + \dots + n^{c-1} = \frac{n^c-1}{n-1} \le n^{c}.
\end{equation*}
Since each $U_n$ uses at most $n^c$ ancillae,
this circuit has at most $n^{2c}$ ancillae.
And the worst case probability of computing $\Parity{n^c}$ is $\br{1+\delta_n^{n^{c}}}/2$,
which is negligibly close to $1$ if $\delta_n = 1 - \negl(n)$.
Formally, we prove the first step of our construction with the following lemma.
\begin{lemma}\label{lem:boost-step-1}
    For $n\ge 3$,
    suppose we have a $\QACz$ circuit which is a $\circsize{d}{n}{n^c}$ circuit,
    and computes $\Parityn$ with the worst-case probability $\br{1+\delta}/2$.
    Then we can construct a $\QACz$ circuit $U^\prime$ such that
    the circuit $U^\prime$ is a $\circsize{c(d+1)}{n^c}{n^{2c}}$ circuit,
    and computes $\Parity{n^c}$ with the worst-case probability $\br{1+\delta^{cn^{c-1}}}/2$.
\end{lemma}
\begin{proof}
    We use \cref{lem:qac0-boost} for $c-1$ times to construct the circuit in \cref{fig:boost-step-1} in a top-to-down manner.
    Initially, we let the circuit $V_1 = U$.
    Then, for $i=2, \dots, c$, we apply \cref{lem:qac0-boost} with $U_b\gets U$ and $U_t\gets V_{i-1}$,
    and let the resulting $\QACz$ circuit be $V_i$.
    Finally, we will let $U^\prime = V_c$.
    
    By induction, we prove that each $V_i$ is a $\circsize{i(d+1)}{n^i}{n^{i+c}}$ circuit that computes $\Parity{n^i}$ with probability $\br{1+\delta^{in^{(c-1)}}}/2$.
    For the base case, $V_1 = U$ is a $\circsize{d}{n}{n^c}$ circuit,
    and computes $\Parityn$ with worst case probability $\br{1+\delta}/2 \ge \br{1+\delta^{c-1}}/2$.
    Then for each $i=2, \dots, c$, by \cref{lem:qac0-boost}, the circuit $V_i$ satisfies:
    \begin{itemize}
        \item depth: $(i-1)(d+1) + d + 1 = i(d+1)$.
        \item input: $n^{i-1}\cdot n = n^i$.
        \item ancillae: $n^{i-1}\br{n^c+1}+n^{i-1+c} \le 3n^{i-1+c} \le n^{i+c}$.
        \item probability: $\br{1 + \delta^{(i-1)n^{(c-1)}}\delta^{i-1}}/2 \ge \br{1 + \delta^{in^{(c-1)}}}/2$.
    \end{itemize}

    The final circuit $V_c$ will be a $\circsize{c(d+1)}{n^c}{n^{2c}}$ circuit,
    that computes $\Parity{n^c}$ with probability $\br{1+\delta^{cn^{(c-1)}}}/2$.
\end{proof}

\paragraph{Step 2}
In the second step, we push the ancillae rate from $n^2$ to $n^{1+\exp\br{-o(d)}}$,
using the construction in \cref{fig:boost-step-2}.
\begin{figure}
    \centering
    \begin{tikzpicture}
        \node[circnode] (T0) at (10,5) {$U_{n^{2^{k-1}}}$};
        \node[circnode] (T11) at (7.5,4) {$U_{n^{2^{k-2}}}$};
        \node[circnode] (T12) at (12.5,4) {$U_{n^{2^{k-2}}}$};
        \node[circnode] (T21) at (5,3) {$U_{n^{2^{k-3}}}$};
        \node[circnode] (T22) at (9,3) {$U_{n^{2^{k-3}}}$};
        \node[circnode] (T23) at (15,3) {$U_{n^{2^{k-3}}}$};
        \node[circnode] (Tn1) at (4,1) {$U_{n}$};
        \node[circnode] (Tn1) at (10,1) {$U_{n}$};
        \node[circnode] (Tn1) at (16,1) {$U_{n}$};

        \draw[->] (T11) -- (T0);
        \draw[->] (T12) -- (T0);
        \draw[->] (T21) -- (T11);
        \draw[->] (T22) -- (T11);
        \draw[->] (T23) -- (T12);

        \node[] at (10,4) {$\cdots$};
        \node[] at (7,3) {$\cdots$};
        \node[] at (12,3) {$\cdots$};
        \node[] at (7,1) {$\cdots$};
        \node[] at (13,1) {$\cdots$};
        \node[] at (6,2) {$\vdots$};
        \node[] at (10,2) {$\vdots$};
        \node[] at (14,2) {$\vdots$};
        \path[draw,decorate,decoration=brace] (2.5,1) -- (2.5,6) node[midway,left]{$k$ layers};
    \end{tikzpicture}
    \Description{A tree composing of unitaries}
    \caption{Step 2, $n^2 \to n^{1+\exp\br{-o(d)}}$ ancillae}
    \label{fig:boost-step-2}
\end{figure}
The construction is similar to the construction \cref{fig:boost-step-1} in step 1.
The difference is that in step 2,
we work with $\QACz$ circuit families instead of a single $\QACz$ circuit.
The reason we need a circuit family is that we will use different $U_n$ in different layers.
Specifically, for each $n$,
for a $k$ layer construction, we let the top layer circuit have $n^{2^{k-1}}$ inputs,
and then the next layer have $n^{2^{k-1}}$ instances of circuits with $n^{2^{k-2}}$ inputs,
and so on.
The bottom layer are circuits with $n^{2^{0}} = n$ inputs.

By using \cref{lem:qac0-boost},
the depth of this construction is $kd + k - 1$,
and it computes the function $\Parity{N}$ for $$N=\prod_{i=0}^{k-1}n^{2^{i}} = n^{2^k-1}.$$
Also, we can check that the ancillae used in each layer is exactly $n^{2^{k}} \approx N^{1+2^{-k}}$.
So the overall ancillae size is $kn^{2^k} \approx N^{1+2^{-k}}$, with the depth being $O(kd)$.
We give this construction formally using the following lemma.
To keep the constructed circuit compute $\Parity{N}$ with non-trivial probability,
we require the circuits to have worst-case correct probability $1 - \negl(n)$.
\begin{lemma}\label{lem:boost-step-2}
    Let $k\ge 2$ be any constant integer.
    Suppose there exists a constant $d\in\mathbb{Z}^+$
    and an integer $n$ such that for each $0 \le i\le k-1$,
    there exist a $\QACz$ circuit $U_i$ that is a $\circsize{d}{n^{2^i}}{n^{2^{i+1}}}$ circuit,
    and computes $\Parity{n^{2^i}}$ with worst case probability $1-\negl(n)$.
    Then we can construct a $\circsize{k(d+1)}{n^{2^k-1}}{2kn^{2^k}}$ circuit
    and computes $\Parity{n^{2^k-1}}$ with worst case probability $1-\negl(n)$.
    
    In particular, we have $2kn^{2^k} = n^{2^k + \log_n\br{2k}}$ and for $n\ge 4k^2$,
    \begin{equation*}
        \frac{2^k + \log_n\br{2k}}{2^k-1} = 1 + \frac{1 + \log_n\br{2k}}{2^k-1} \le 1 + \frac{2 - 2^{-k+1}}{2^k-1} = 1 + 2^{-k+1}.
    \end{equation*}
    So we get a $\circsize{k(d+1)}{N}{N^{1+2^{-k+1}}}$ circuit for $N=n^{2^k-1}$, with depth $k(d+1)$.
\end{lemma}
\begin{proof}
    We use \cref{lem:qac0-boost} for $k-1$ times to construct the circuit in \cref{fig:boost-step-2} in a down-to-top manner.
    We will use $V_i$ to represent a sub-tree in \cref{fig:boost-step-2} with $i$ layers.
    Initially, we let $V_1 = U_{0}$.
    Then for each $i=2, \dots, k$,
    we apply \cref{lem:qac0-boost} with $U_b\gets V_{i-1}$ and $U_t\gets U_{i-1}$,
    and let the resulting $\QACz$ circuit be $V_i$.
    Remember by assumption, $U_{i-1}$ is a $\circsize{d}{n^{2^{i-1}}}{n^{2^i}}$ circuit.
    We will prove by induction that each $V_i$ is a $\circsize{i(d+1)}{n^{2^i-1}}{2in^{2^i}}$ circuit,
    that computes $\Parity{n^{2^{i}-1}}$ with worst-case probability $1-\negl(n)$.

    For the base case, $V_1 = U_0$ is a $\circsize{d}{n}{n^2}$ circuit.
    Now for each $i=2, \dots, k$, by \cref{lem:qac0-boost}, the circuit $V_i$ satisfies
    \begin{itemize}
        \item depth: $(i-1)(d+1) + d + 1 \le i (d+1)$.
        \item input: $n^{2^{i-1}-1}\cdot n^{2^{i-1}} = n^{2^i-1}$.
        \item ancillae: $n^{2^{i-1}}\br{2(i-1)n^{2^{i-1}}+1}+n^{2^{i}} = 2(i-1)n^{2^{i}}+n^{2^{i-1}}+n^{2^{i}} \le 2in^{2^{i}}$.
    \end{itemize}
    Finally, since $k$ is a constant,
    the above construction only uses polynomial instances of the circuits $U_i$,
    so the circuit $V_k$ computes $\Parity{n^{2^{k}-1}}$ with worst-case probability $1-\negl(n)$.

%
\end{proof}

\paragraph{Putting together}
With \cref{lem:boost-step-1} and \cref{lem:boost-step-2},
we can directly prove \cref{thm:boost-full},
which is restated below:
\restateboost*
\begin{proof}
    Let $D = 3c(d+1)$.
    Let $n$ be any integer satisfying
    \begin{itemize}
        \item $n=m^{\br{2^{K+1}-1}c}$ for some integer $m \ge N_0$.
        \item $m^c\ge 4K^2$.
    \end{itemize}
    For $i=0, \dots, K$, apply \cref{lem:boost-step-1} to the circuit $U_{m^{2^i}}$ to get the $\circsize{c(d+1)}{m^{2^ic}}{m^{2^{i+1}c}}$ circuit $U^\prime_i$.
    Then apply \cref{lem:boost-step-2} with $k\gets K+1$ and $n\gets m^c$ to the circuits $\set{U^\prime_i}_{i\in\set{0,\dots,K}}$.
    We get a $\circsize{(K+1)(c(d+1)+1)}{n}{n^{1+2^{-K}}}$ circuit.
    Then we finish the proof by showing
    \begin{equation*}
        (K+1)(c(d+1)+1) \le KD.
    \end{equation*}
\end{proof}

\section{Quantum State Synthesis}

In this section, we investigate the hardness of quantum state synthesis. Rosenthal~\cite{rosenthal:LIPIcs.ITCS.2021.32} has shown that the cat state synthesis and \parity\ are equivalent via \QACz\ reduction. Thus, the hardness of cat state synthesis can be derived from \Cref{thm:main:WorstCase}. This section will provide a more generic method to prove the hardness of state synthesis via  low-degree property of \QACz\ circuits in \cref{cor:qac0-whole}.

For a quantum state $\varphi$, its normalized Frobenius norm is exponentially small.
\begin{equation*}
    \normsub{\varphi}{2} = \br{2^{-n}\Tr\br{\varphi^2}}^{1/2} \le 2^{-n/2}.
\end{equation*}
Hence if we approximate quantum states with constant Frobenius distance,
we can always get trivial approximation results.
However, the spectral norm of a quantum state can be as large as $1$ for pure states.
This means that the approximation results of \cite{nadimpalli_pauli_2024} are insufficient,
and spectral approximate degree is essential in the task of quantum state synthesis hardness.


\subsection{Lower Bound on the Degrees of Quantum States}\label{sec:degree-lower-bound-of-quantum-states}
We first study the approximate degree of quantum states.
For a quantum state $\varphi$,
we continue to use the approximate degree $\degeps{\varepsilon}{\varphi}$ defined in \cref{def:quantum-approximate-degree}.
For low-degree pure quantum states, we have the following concentration bound when we measure it in the computational basis:
\begin{lemma}[\cite{Anshu2023concentrationbounds,KAAV15}]\label{lem:anticoncentration}
    Let $\varphi = \ketbra{\varphi}$ be a pure state satisfying
    $\degeps{\varepsilon}{\varphi} \leq k$.
    We measure $\varphi$ in the computational basis,
    and let $W_\varphi$ denote the Hamming weight of the measurement outcome.
    That is, we let $W_\varphi$ be the random variable satisfying
    $\prob{W_\varphi = i} = \sum_{x\in\set{0,1}^n: \abs{x}=i}\bra{x}\varphi\ket{x}$.
    Let $m$ be an integer median of $W_\varphi$ that satisfies
    $\prob{W_\varphi \le m} \ge 1/2$,
    then
    \begin{equation}\label{eq:quamtum-state-concentration1}
        \prob{W_\varphi > m + k} \le 4\varepsilon^2.
    \end{equation}
    Similarly, if $m$ is an integer median that satisfies
    $\prob{W_\varphi \ge m} \ge 1/2$,
    then
    \begin{equation}\label{eq:quamtum-state-concentration2}
        \prob{W_\varphi < m - k} \le 4\varepsilon^2.
    \end{equation}
\end{lemma}
Lemma~\ref{lem:anticoncentration} implies that the quantum states that do not satisfy this concentration property actually have high approximate degree.
This allows us to prove a tight bound on the approximate degree of a nekomata state. 
\begin{cor}\label{cor:degnekomata}
    For $\varepsilon < \frac{1}{4\sqrt{2}}$ and $n$-nekomata state $\ket{\nu} = \f{1}{\sqrt{2}} \p{\ket{0^n,\psi_0} + \ket{1^n,\psi_1}}$, we have
    \begin{equation*}
        \degeps{\varepsilon}{\ketbra{\nu}} \ge n.
    \end{equation*}
    As a special case, for cat states, we have
    \begin{equation*}
        \degeps{\varepsilon}{\Cat_n} \ge n.
    \end{equation*}
\end{cor}
\begin{proof}
    We first prove the special case for cat states.
    Let $W_{\Cat}$ be the random variable defined in \cref{lem:anticoncentration}. It is easy to see that 
    $$W_{\Cat_n} = \begin{cases}
        n \text{ with probability } 1/2.\\
        0 \text{ with probability } 1/2.
    \end{cases}$$
    Applying \cref{lem:anticoncentration} with $k\leftarrow n-1$ and $m\leftarrow 0$, we conclude the result $\degeps{\varepsilon}{\Cat_n} \ge n$.
    
    Now for a general $n$-nekomata, assume $\braket{\psi_0}{\psi_1} \ge 0$. Let 
    $$\tilde{\nu} = \Tr_{\ge n+1}\Br{\br{\id\otimes \f{\ketbra{\psi_0} + \ketbra{\psi_1} + \ketbra{\psi_0}{\psi_1} + \ketbra{\psi_1}{\psi_0}}{\Tr\br{ \ketbra{\psi_0} + \ketbra{\psi_1} + \ketbra{\psi_0}{\psi_1} + \ketbra{\psi_1}{\psi_0} } }} \ketbra{\nu}}.$$ 
    With  \cref{lem:zeropart}, $\degeps{\varepsilon}{\tilde{\nu}} \le \degeps{\varepsilon}{\nu}$.
    
    We finish our proof with the fact \[\tilde{\nu} = \p{\frac{|1 + \braket{\psi_0}{\psi_1}|^2}{\Tr\br{ \ketbra{\psi_0} + \ketbra{\psi_1} + \ketbra{\psi_0}{\psi_1} + \ketbra{\psi_1}{\psi_0} } }}\Cat_n.\]
\end{proof}

\subsection{Quantum State Synthesis}
Now we can prove our hardness results for quantum state synthesis.
The formal definition of quantum state synthesis is \cref{def:state-synthesis}.
This definition coincides with the approximate dirty state generation problem considered by Rosenthal~\cite{rosenthal:LIPIcs.ITCS.2021.32},
and is the easiest one among his definitions.

We will be using the following lemma, with the proof deferred to \cref{appendix:proofs}.
\begin{lemma}\label{lem:pure-state-purification}
    Let $\varphi = \ketbra{\varphi}$ be a pure state on $n$ qubits.
    Let $\psi = \ketbra{\psi}$ be a pure state on $n+a$ qubits.
    Suppose $$\norm{\ketbra{\varphi} - \Tr_{n+1,\dots,n+a}\Br{\ketbra{\psi}}} \le \varepsilon,$$
    then there exists a pure quantum state $\ket{\nu}$ on $a$ qubits such that
    $$\norm{\ketbra{\varphi}\otimes\ketbra{\nu} - \ketbra{\psi}} \le 5\sqrt{\varepsilon}.$$
\end{lemma}

\begin{theorem}\label{thm:state-synthesis}
    Let $\varphi = \ketbra{\varphi}$ be a pure state on $n$ qubits.
    Let $U$ be a $\QACz$ circuit with $a$ ancillae and depth $d$ that synthesizes $\varphi$ with fidelity $1 - \delta$.
    Then for $\varepsilon = 10\delta^{1/4} + \Omega(d/n) $, we have
    \begin{equation*}
        (n+a)^{1-3^{-d}/2} = \tilde{\Omega}\br{\degeps{\varepsilon}{\varphi}}.
    \end{equation*}
    In particular, for the $n$-nekomata state, and for some fidelity $1 - \delta_0 < 0$, we have
    \begin{equation*}
        (n+a)^{1-3^{-d}/2} = \tilde{\Omega}(n).
    \end{equation*}
\end{theorem}

Setting $a=O(n)$ and combing with \cref{cor:degnekomata}, we have the following corollary.
\begin{cor}
For any $0 < \delta\le 1$, it holds that 
   \[\degeps{\varepsilon}{\stateQLCz\Br{\delta}}=o(n).\]
   for $\varepsilon = 10\delta^{1/4} + \Omega(d/n)$. In particular, for any family of $n$-nekomata state $\set{\ketbra{\nu_n}}_{n\in\mathbb{N}}$, it holds that $\set{\ketbra{\nu_n}}_{n\in\mathbb{N}}\notin\stateQLCz$.
\end{cor}
\begin{remark}
    Our \cref{def:state-synthesis} requires that the output of the $\QACz$ circuit have fidelity $1-\delta$ with the target state $\psi$,
    which is a relatively strong assumption.
    In the proof of \cref{thm:state-synthesis},
    to prove our circuit size lower bound,
    we actually only require the output state of the $\QACz$ circuit to be close to the target state $\psi$
    in spectral norm.
    This is a very weak assumption because, according to the Fuchs–van de Graaf inequalities~\cite[Theorem 3.33]{watrous2018theory},
    high fidelity is equivalent to low unnormalized trace norm, and the spectral norm can be exponentially smaller than the unnormalized trace norm.
    Hence, the first step of the proof would be acquiring a spectral norm upper bound given the fidelity requirement.
\end{remark}
\begin{proof}[Proof of \cref{thm:state-synthesis}]
  By \cref{def:state-synthesis}, we have
  \begin{equation*}
    F\br{\ketbra{\varphi}, \Tr_{n+1,\dots,n+a}\Br{U\ketbra{0}^{n+a}U^\dagger}} = \sqrt{\bra{\varphi}\Tr_{n+1,\dots,n+a}\Br{U\ketbra{0}^{n+a}U^\dagger}\ket{\varphi}} \ge 1 - \delta.
  \end{equation*}
  By the Fuchs–van de Graaf inequalities~\cref{lem:fuchs-vandegraaf},
  along with the fact that the unnormalized trace norm is larger than the spectral norm, we have
  \begin{equation*}
    \norm{\varphi - \Tr_{n+1,\dots,n+a}\Br{U\ketbra{0}^{n+a}U^\dagger}} \le 2\sqrt{\delta}.
  \end{equation*}
  Now by \cref{lem:pure-state-purification},
  there exists a pure state $\nu = \ketbra{\nu}$ on $a$ qubits such that
  \begin{equation*}
    \norm{\varphi\otimes\nu - U\ketbra{0}^{n+a}U^\dagger} \le 10\delta^{1/4}.
  \end{equation*}
  We turn to prove a low-degree approximation of the operator $U\ketbra{0}^{n+a}U^\dagger$.
  First by \cref{lem:locality-of-product}, we have
  \begin{equation*}
    \degeps{O(1/(n+a))}{\ketbra{0}^{n+a}} = \tilde{O}\br{\sqrt{n+a}}.
  \end{equation*}
  Hence for $\tau = O(d/n)$, by \cref{cor:qac0-whole},
  we have $\degeps{\tau}{U\ketbra{0}^{n+a}U^\dagger} \le \tilde{O}\br{(n+a)^{1-3^{-d}/2}}$.
  This implies that the state $\varphi\otimes\nu$ has a degree $\tilde{O}\br{(n+a)^{1-3^{-d}/2}}$ approximation
  with spectral distance at most $10\delta^{1/4} + O(d/n) \le \varepsilon$.
  Finally, the lemma follows because
  \begin{equation*}
    \degeps{\varepsilon}{\varphi} \le \degeps{\varepsilon}{\varphi\otimes\nu}.
  \end{equation*}
  Now we prove the circuit lower bound for synthesizing $n$-nekomata states.
  By \cref{cor:degnekomata}, we need $\varepsilon < \frac{1}{4\sqrt{2}}$.
  Hence we only need to set $\delta_0$ such that $10\delta_0^{1/4} < \varepsilon < \frac{1}{4\sqrt{2}}$,
  which implies $\delta_0 < \frac{1}{10240000}.$
\end{proof}

\subsection{Synthesizing Long Range Correlation Using \texorpdfstring{$\QACz$}{QAC0} Circuits}

It is well known that the output state $\theta$ of a shallow quantum circuit - with a product state input - has zero correlation length. That is, the ``two-point correlation function'' $\Tr(AB\theta) - \Tr(A\theta)\Tr(B\theta)$ is $0$ for all local operators $A,B$ that are supported out of each other's light cones. In this section, we give an example which shows that $\QACz$ circuits are more powerful and can produce long range entanglement in depth $1$.

We start with the state $$\ket{\rho^0}=\left(\sqrt{1-\frac{1}{n}}\ket{0}+\sqrt{\frac{1}{n}}\ket{1}\right)^{\otimes n}=\sum_{k=0}^n a_k \ket{\psi_k},$$
where $a_k=\sqrt{\binom{n}{k}\left(1-\frac{1}{n}\right)^{n}\frac{1}{(n-1)^k}}$ and $\ket{\psi_k}= \frac{1}{\sqrt{\binom{n}{k}}}\sum_{x:|x|=k}\ket{x}$ is the uniform superposition over strings of Hamming weight $k$. The following state in $\QACz$: 
$$\ket{\rho^1}=X^{\otimes n}\CZG X^{\otimes n}\ket{\rho^0}=-a_0\ket{\psi_0} + \sum_{k=1}^na_k\ket{\psi_k}$$ will be shown to have a long range correlation. Towards this, define two projectors for a subset $S$ of qubits: 
$$\Pi^0_S = \ketbra{0}^{\otimes |S|},\quad \Pi^1_S = \frac{1}{2}\left(\ket{0}^{\otimes |S|} + \frac{1}{\sqrt{|S|}}\sum_{x_S:|x_S|=1}\ket{x_S}\right)\left(\bra{0}^{\otimes |S|} + \frac{1}{\sqrt{|S|}}\sum_{x_S:|x_S|=1}\bra{x_S}\right).$$
We have 
$$\Pi^0_S\ket{\rho^0}=a_0\ket{\psi_0} + \ket{0}^{|S|}\otimes\sum_{k=1}^na_k\ket{\psi_k^{S_c}}, \quad \Pi^0_S\ket{\rho^1}=-a_0\ket{\psi_0} + \ket{0}^{|S|}\otimes\sum_{k=1}^na_k\ket{\psi_k^{S_c}},$$
where $\ket{\psi_k^{S_c}} = \frac{1}{\sqrt{\binom{n}{k}}}\sum_{x_{S_c}:|x_{S_c}|=k}\ket{x_{S_c}}$ is a sub-normalized state on the qubits in $S_c$ (the complement of $S$) with Hamming weight $k$. Note that $$\|\Pi^0_S\ket{\rho^0}\|=\|\Pi^0_S\ket{\rho^1}\| = \left(1-\frac{1}{n} \right)^{|S|},$$ since the minus sign in front of $\psi_0$ does not affect the norm and the last equality is the probability that $|S|$ sequences of $0$ are seen. Next, 
\begin{eqnarray*}
\Pi^1_S\ket{\rho^0}&=&\frac{1}{2}\left(\ket{0}^{\otimes |S|} + \frac{1}{\sqrt{|S|}}\sum_{x_S:|x_S|=1}\ket{x_S}\right)\left(a_0\ket{0}^{\otimes |S_c|}+\sum_{k=1}^na_k\ket{\psi_k^{S_c}}+a_1\sqrt{\frac{|S|}{n}}\ket{0}^{\otimes |S_c|}+\ket{\mu}\right)\\
&=& \frac{1}{2}\left(\ket{0}^{\otimes |S|} + \frac{1}{\sqrt{|S|}}\sum_{x_S:|x_S|=1}\ket{x_S}\right)\left(\left(a_0+a_1\sqrt{\frac{|S|}{n}}\right)\ket{0}^{\otimes |S_c|}+\sum_{k=1}^na_k\ket{\psi_k^{S_c}}+\ket{\mu}\right) \\
\Pi^1_S\ket{\rho^1}&=&\frac{1}{2}\left(\ket{0}^{\otimes |S|} + \frac{1}{\sqrt{|S|}}\sum_{x_S:|x_S|=1}\ket{x_S}\right)\left(-a_0\ket{0}^{\otimes |S_c|}+\sum_{k=1}^na_k\ket{\psi_k^{S_c}}+a_1\sqrt{\frac{|S|}{n}}\ket{0}^{\otimes |S_c|}+\ket{\mu}\right)\\
&=& \frac{1}{2}\left(\ket{0}^{\otimes |S|} + \frac{1}{\sqrt{|S|}}\sum_{x_S:|x_S|=1}\ket{x_S}\right)\left(\left(-a_0+a_1\sqrt{\frac{|S|}{n}}\right)\ket{0}^{\otimes |S_c|}+\sum_{k=1}^na_k\ket{\psi_k^{S_c}}+\ket{\mu}\right) 
\end{eqnarray*}
where $\ket{\mu}:=\left(\frac{1}{\sqrt{|S|}}\sum_{x_S:|x_S|=1}\bra{x_S}\right)\left(\sum_{k=2}^n a_k\ket{\psi_k}\right)$ is a state with Hamming weight at least $1$ on the qubits in $S_c$. Note that
\begin{equation}
\label{eq:expecdiff1}
\|\Pi^1_S\ket{\rho^0}\|^2 - \|\Pi^1_S\ket{\rho^1}\|^2 = \frac{\left(a_0+a_1\sqrt{\frac{|S|}{n}}\right)^2-\left(-a_0+a_1\sqrt{\frac{|S|}{n}}\right)^2}{2} = 2a_0a_1\sqrt{\frac{|S|}{n}}.
\end{equation}
From here, we can compare the correlation functions 
$$C_0:= \bra{\rho^0}\Pi^0_S \Pi^1_T\ket{\rho^0} -\bra{\rho^0}\Pi^0_S\ket{\rho^0} \bra{\rho^0}\Pi^1_T\ket{\rho^0}, \quad C_1:= \bra{\rho^1}\Pi^0_S \Pi^1_T\ket{\rho^1} -\bra{\rho^1}\Pi^0_S\ket{\rho^1} \bra{\rho^1}\Pi^1_T\ket{\rho^1},$$ for two distinct sets of qubits $S,T$. The difference  $$\bra{\rho^0}\Pi^0_S\Pi^1_T\ket{\rho^0}-\bra{\rho^1}\Pi^0_S\Pi^1_T\ket{\rho^1} = \|\Pi^1_T\Pi^0_S\ket{\rho^0}\|^2 - \|\Pi^1_T\Pi^0_S\ket{\rho^1}\|^2,$$ can be evaluated similar to Equation \ref{eq:expecdiff1} by replacing the coefficients $a_k$ with $a_k\sqrt{\frac{\binom{n-|S|}{k}}{\binom{n}{k}}}$ (since, for each $b\in\{0,1\}$, the states $\ket{\rho^b}$ and $\Pi^0_S\ket{\rho^b}$ are very similar except that the former is defined on $n$ qubits, latter is defined on $n-|S|$ qubits and the coefficients are rescaled). This shows that $$\bra{\rho^0}\Pi^0_S\Pi^1_T\ket{\rho^0}-\bra{\rho^1}\Pi^0_S\Pi^1_T\ket{\rho^1} = 2a_0a_1\sqrt{\frac{n-|S|}{n}}\sqrt{\frac{|T|}{n}}.$$ Collectively, using the fact that $C_0=0$ for the product state $\rho^0$, we have
\begin{eqnarray*}
-C_1&=&C_0-C_1= 2a_0a_1\sqrt{\frac{n-|S|}{n}}\sqrt{\frac{|T|}{n}} - 2a_0a_1\sqrt{\frac{|S|}{n}}\cdot\|\Pi^0_S\ket{\rho^0}\|^2 \\
&=& 2a_0a_1\left(\sqrt{\frac{n-|S|}{n}}\sqrt{\frac{|T|}{n}}-\sqrt{\frac{|S|}{n}}\left(1-\frac{1}{n}\right)^{|S|}\right).
\end{eqnarray*}
We can choose $|S|=|T|= \Theta(n)$, which gives $|C_1|= \Theta(1)$ using $a_0,a_1=\Theta(1)$. This is near maximal correlation between two large regions and if we are considering lattices, one can arrange for these regions to be far apart. If we insist on considering operators of constant locality, we can choose $|S|=1, |T|=2$ and obtain $|C_1| = \Theta(1/\sqrt{n})$. 

\subsection{Bounds on Low Energy State Preparation}

Given that $\QACz$ circuits can produce long range correlations, the possibility of using them to probe low-energy regime of local Hamiltonians emerges. More concretely, we can ask if a sufficiently low-energy state of any local Hamiltonian can be prepared by a $\QACz$ circuit. Here we argue that this is not possible - 
 \cref{cor:qac0-whole} implies that the low-energy states of the local Hamiltonian considered in \cite{anshu2022nlts} cannot be generated by $\QACz$ circuits.
This holds since we can use a concentration bound similar to \cref{lem:anticoncentration} to show approximate degree lower bounds on the low-energy states of~\cite{anshu2022nlts}.

\begin{lemma}
    Consider a $[[n, k, d]]$ CSS code satisfying Property $1$ from \cite{anshu2022nlts} with parameters $\delta_0, c_1, c_2$ as stated.
    Let $\mathbf{H} = \mathbf{H}_x + \mathbf{H}_z$ be the corresponding local Hamiltonian,
    with $m_x$ and $m_z$ local terms, respectively.
    Then for
    $$\varepsilon < \frac{1}{400c_1}\br{\frac{\min\set{m_x, m_z}}{n}}\cdot\min\set{\br{\frac{k-1}{4n}}^2, \delta_0, \frac{c_2}{2}},$$
    and every pure state $\varphi = \ketbra{\varphi}$ such that $\Tr\Br{\mathbf{H}\varphi}\le\varepsilon n$,
    we have for $\delta = 1/8000$,
    \begin{equation*}
        \degeps{\delta}{\varphi} = \Omega(n).
    \end{equation*}
    In particular, the quantum Tanner codes~\cite{9996782} satisfy the above property.
\end{lemma}
\begin{proof}
    Let $W_z$ be the measurement outcome of $\varphi$ on the computational basis,
    and $W_x$ be the measurement outcome of $\varphi$ on the Hadamard basis,
    which is equivalent to measuring $H^{\otimes n}\varphi H^{\otimes n}$ on the computational basis.
    By \cite[Lemma 2]{anshu2022nlts}, there exist $b\in\set{x, z}$ and two sets $B_0$ and $B_1$ such that
    \begin{equation*}
        W_b(B_0) \ge \frac{1}{400} \text{ and } W_b(B_1) \ge \frac{1}{400}.
    \end{equation*}
    Also, the distance between $B_0$ and $B_1$ is $\Omega(n)$.
    
    Afterward, if $b = z$, then we will prove the approximate degree lower bounds for $\varphi$,
    and if $b = x$, we will prove the approximate degree lower bounds for $H^{\otimes n}\varphi H^{\otimes n}$.
    Since $\degeps{\delta}{\varphi} = \degeps{\delta}{H^{\otimes n}\varphi H^{\otimes n}}$,
    we can assume without loss of generality that $b = z$.
    
    Now we can use an argument similar to \cref{lem:anticoncentration}, to prove that $\degeps{\delta}{\varphi} = \Omega(n)$.
    Let $\Pi_{0}$ and $\Pi_{1}$ be the projectors on the strings in $B_0$ and $B_1$ respectively.
    For any operator $R$ with degree smaller than $\operatorname{dist}(B_0, B_1) = \Omega(n)$, we have
    $\Pi_0R\Pi_1 = 0$. Then
    \begin{align*}
        \norm{\varphi - R} &\ge \norm{\Pi_0\br{\varphi - R}\Pi_1} \\
                    &= \norm{\Pi_0\varphi\Pi_1} \\
                    &\ge \norm{\Pi_0\ketbra{\varphi}\Pi_1\ket{\varphi}}_2 \\
                    &= \bra{\varphi}\Pi_1\ket{\varphi}\normsub{\Pi_0\ket{\varphi}}{2} \\
                    &= \Tr\Br{\Pi_1\varphi}\cdot\br{\Tr\Br{\Pi_0\varphi}}^{1/2} \\
                    &= W_0(B_1)\cdot\sqrt{W_0(B_0)} \\
                    &\ge \frac{1}{8000}.
    \end{align*}
    Thus every operator $R$ with degree lower than $\Omega(n)$ has spectral distance at least $\frac{1}{8000}$ to $\varphi$.
    Hence we have
    $$\degeps{1/8000}{\varphi} = \Omega(n).$$
\end{proof}

With the approximate degree lower bound above,
we can immediately invoke \cref{thm:state-synthesis} to get the following hardness result for $\QACz$ circuits.
\begin{cor}
    There exists a constant $\varepsilon_0 > 0$, such that
    for all depth-$d$ $\QACz$ circuits with $a$ ancillae
that synthesize the low energy state $\varphi$ with fidelity $1 - \delta = 1-\varepsilon_0 + O(d^4/n^4)$,
    we have
    \begin{equation*}
        a = \tilde{\Omega}\br{n^{1+3^{-d}/2}}.
    \end{equation*}
\end{cor}
\begin{proof}
    By \cref{thm:state-synthesis}, we have for
    \begin{equation*}
    \varepsilon = 10\delta^{1/4} + \Omega(d/n) = 10\br{\varepsilon_0 - O(d^2/n^2)}^{1/4} + \Omega(d/n) = 10\varepsilon_0^{1/4},
    \end{equation*}
    \begin{equation*}
        \br{n + a}^{1-3^{-d}/2} = \tilde{\Omega}\br{\degeps{\varepsilon}{\varphi}}.
    \end{equation*}
    For $\varepsilon = \frac{1}{8000}$, we have $\degeps{\varepsilon}{\varphi} = \Omega(n)$.
    Thus for $\varepsilon_0 = \frac{\varepsilon^4}{10000}$,
    we have
    \begin{equation*}
        \br{n + a}^{1-3^{-d}/2} = \tilde{\Omega}\br{n}.
    \end{equation*}
\end{proof}


\section{Quantum Channels Synthesis}


In this section, we prove \QACz\ hardness results for general quantum channels.
Recall for a unitary operator $U$, we use $\CE_{k, U, \psi}$ to define the $k$ qubit output quantum channel
using $\psi$ as ancilla as
$$\CE_{k, U, \psi}(\rho) = \Tr_{[k]^c}\Br{U\br{\rho\otimes \psi}U^\dagger}.$$
The Choi representation of $\CE_{k, U, \psi}$ is denoted by
$$\Phi_{k, U, \psi} = \br{\CE_{k, U, \psi}\otimes\id}\br{\EPR_n},$$
where $\EPR_n$ is the density operator of unnormalized $n$-qubits EPR state $\sum_{x\in\set{0,1}^n}\ket{x}\otimes\ket{x}$.
The subscript $k$ may be omitted if it is clear from context.
If there is no ancilla, we use the notation $\CE_U$ and $\Phi_U$.
We adapt the following identity for the Choi state of quantum channels.
\begin{fact}[\cite{nadimpalli_pauli_2024}]\label{fact:Choi:1}
    $$\Phi_U = \br{\id\otimes U^T}\br{\EPR_k\otimes\id_{n-k}}\br{\id\otimes\overline{U}}.$$
    If $\psi = \ketbra{\psi}$ is a pure state,
    $$\Phi_{U, \psi} = \bra{\psi}\Phi_{U}\ket{\psi}.$$
\end{fact}
Since we are working with spectral norm approximations of matrices,
the spectral norms of the Choi states are important.
In fact, by \cref{fact:Choi:1} and \cite[Lemma 4]{10510479},
we see that the spectral norms of Choi states are upper bounded by $2^k$.
We approximate the operator $2^{-k}\Phi_{U, \psi}$,
which is achieved by using \cref{cor:qac0-whole}.
\begin{theorem}\label{thm:qchannel-degree}
    Let $n\ge 1$.
    Suppose $U$ is a depth-$d$ \QACz\ circuit with $n$ input qubits and $a$ ancillae initialized in the state $\psi$.
    For $k\le n$, we take the first $k$ qubits as output and implement the quantum channel $\CE_{U, \psi}$.
    For $\ell = \tilde{O}\br{(n+a)^{1-3^{-D}}k^{3^{-D}/2}}, \varepsilon = O(d/n)$, we have
    \begin{equation}
        \degeps{\varepsilon}{2^{-k}\Phi_{U, \psi}} \le \ell.
    \end{equation}
\end{theorem}
\begin{remark}
    This result is incomparable to the result of \cite{nadimpalli_pauli_2024}.
    In our result, we approximate the channel by an $o(n)$-degree operator with respect to the spectral norm.
    In \cite{nadimpalli_pauli_2024}, they approximate the channel by a constantly local operator with respect to the Frobenius norm.
\end{remark}
\begin{proof}[Proof of \cref{thm:qchannel-degree}]
    By \cref{fact:Choi:1},
    $$2^{-k}\Phi_{U, \psi} = 2^{-k}\bra{\psi}\Phi_{U}\ket{\psi} = \bra{\psi}\br{\id\otimes U^T}\br{2^{-k}\EPR_k\otimes\id_{n-k}}\br{\id\otimes\overline{U}}\ket{\psi}.$$
    So the proof idea would be to get a low-degree approximation of $\EPR_k$ and then invoke \cref{cor:qac0-whole}.
    However, note that $\norm{\EPR_k} = 2^k$, so instead we approximate the operator $2^{-k}\EPR_k$.
    Note that $2^{-k}\EPR_k$ is the tensor product of $k$ EPR pairs. Hence by \cref{lem:locality-of-product},
    \begin{equation*}
        \degeps{O(1/n)}{2^{-k}\EPR_k} \le \tilde{O}\br{\sqrt{k}}.
    \end{equation*}
    Then by \cref{cor:qac0-whole}, the theorem follows.
\end{proof}
Using this theorem,
we can prove it is hard for $\QACz$ circuits with linear ancillae to implement high degree quantum channels.
Here, we use the completely bounded spectral norm as a measure for quantum channels.
\begin{definition}[Completely Bounded Spectral Norm]
  Let $\CE$ be a quantum channel. The completely bounded spectral norm of $\CE$ is defined as
  \begin{equation*}
    \norm{\CE} = \max_{X: \norm{X}\le 1}\norm{\br{\CE\otimes\id}(X)}.
  \end{equation*}
\end{definition}
We then have the following corollary of \cref{thm:qchannel-degree}:
\begin{cor}
    Let $\CE$ be a quantum channel from $n$ qubits to $k$ qubits.
    Suppose there exists a \QACz\ circuit $U$ with $a$ ancillae initialized in the state $\psi$ that approximates $\CE$.
    That is, for some $\varepsilon \le 1$,
    \begin{equation*}
        \norm{\CE - \CE_{U, \psi}} \le 2^k\varepsilon,
    \end{equation*}
    then let $\Phi$ be the Choi state of $\CE$, we have
    \begin{equation*}
        \degeps{\varepsilon+O(d/n)}{2^{-k}\Phi} \le \tilde{O}\br{(n+a)^{1-3^{-D}}k^{3^{-D}/2}}.
    \end{equation*}
\end{cor}
\begin{proof}
  By the definition of completely bounded spectral norms,
  \begin{equation*}
      2^k\varepsilon\ge\norm{\CE - \CE_{U, \psi}} = \max_{X: \norm{X}\le 1}\norm{\br{\br{\CE - \CE_{U, \psi}}\otimes\id}(X)} \ge \norm{\Phi - \Phi_{U, \psi}}.
  \end{equation*}
  Then by \cref{thm:qchannel-degree}, for $\varepsilon^\prime = \varepsilon + O(d/n)$, we have
    \begin{equation*}
        \degeps{\varepsilon+O(d/n)}{2^{-k}\Phi} \le \tilde{O}\br{(n+a)^{1-3^{-D}}k^{3^{-D}/2}}.
    \end{equation*} 
\end{proof}

%

\bibliographystyle{alpha}
\bibliography{references}

\appendix
\section{Deferred Proofs}\label{appendix:proofs}

\begin{proof}[Proof of \cref{lem:spectral-multiply}]
    The fact that $\norm{AB}\le1$ follows from the submultiplicativity of the Schatten p-norm \cite[Eq. 1.176]{watrous2018theory}. Then
    \begin{align*}
        \norm{AB - \tilde{A}\tilde{B}} &= \norm{AB - A\tilde{B} + A\tilde{B} - \tilde{A}\tilde{B}} \\
          &\le \norm{A \br{B-\tilde{B}}} + \norm{\br{A-\tilde{A}} \tilde{B}} \\
          &\le \norm{A}\norm{\br{B-\tilde{B}}} + \norm{\br{A-\tilde{A}}}\norm{\tilde{B}} \\
          &\le \varepsilon_1 + \varepsilon_0\norm{\tilde{B} - B + B} \\
          &\le \varepsilon_1 + \varepsilon_0\br{\norm{\tilde{B} - B} + \norm{B}} \\
          &\le \varepsilon_0 + \varepsilon_1 + \varepsilon_0\varepsilon_1.
    \end{align*}
\end{proof}

\begin{proof}[Proof of \cref{thm:ac0-circ-qac0}]
    Let $A = U^\dagger M_fU$.
    Then we have
    $$p(x) = \Tr\Br{\br{\ketbra{x}\otimes\varphi}A} = \bra{x}\Tr_{n + 1,\dots, n + a}\Br{A\br{\id\otimes\varphi}}\ket{x}.$$
    Then the diagonal matrix $M_p$ as defined in~\Cref{eqn:diagonalf} can be obtained by zeroing out all the non-diagonal entries of the matrix
    $$\Tr_{n + 1,\dots, n + a}\Br{A\br{\id\otimes\varphi}}.$$
    By \cref{cor:qac0-whole}, we have $\degeps{\varepsilon}{\Tr_{n+1,\dots, n+a}\Br{A\br{\id\otimes\varphi}}} \le \tilde{O}\br{(n+a)^{1-3^{-d}}\cdot\ell^{3^{-d}}}$.
    Finally, by \cref{lem:diagonaldegree} and \cref{fact:degreecoincide} we prove our theorem.
\end{proof}

\begin{proof}[Proof of \cref{lem:pure-state-purification}]
    First we write down a Schmidt decomposition of $\psi$ as
    \begin{equation*}
        \ket{\psi} = \sum_{i}\sqrt{s_i}\ket{\mu_i}\otimes\ket{\nu_i},
    \end{equation*}
    where $\set{\mu_i}$ is a set of orthogonal basis on $n$ qubits,
    and $\set{\nu_i}$ is a set of orthogonal basis on $a$ qubits.
    Also, we assume $s_1\ge s_2\ge\dots \ge s_{\rank(\psi)}$.
    Then
    \begin{equation*}
        \Tr_{n+1,\dots,n+a}\Br{\ketbra{\psi}} = \sum_is_i\ketbra{\mu_i},
    \end{equation*}
    and
    \begin{align*}
        \varepsilon &\ge \norm{\varphi - \Tr_{n+1,\dots,n+a}\Br{\ketbra{\psi}}} \\
                    &= \norm{\varphi - \sum_is_i\ketbra{\mu_i}} \\
                    &\ge \bra{\varphi}\br{\varphi - \sum_is_i\ketbra{\mu_i}}\ket{\varphi} \\
                    &= 1 - \sum_is_i\abs{\bra{\mu_i}\ket{\varphi}}^2 \\
                    &\ge 1 - \sum_is_1\abs{\bra{\mu_i}\ket{\varphi}}^2 \\
                    &= 1 - s_1.
    \end{align*}
    Hence $s_1 \ge 1 - \varepsilon$, and then $\sum_{i\ge 2}s_i\le 1 - s_1 \le \varepsilon$.
    \begin{equation*}
        \norm{\varphi - \mu_1} \le \norm{\varphi - \sum_is_i\ketbra{\mu_i}} + \norm{(1-s_1)\ketbra{\mu_1} - \sum_{i\ge 2}s_i\ketbra{\mu_i}} \le \varepsilon + \varepsilon = 2\varepsilon.
    \end{equation*}
    \begin{align*}
        \norm{\ket{\psi} - \ket{\mu_1}\otimes\ket{\nu_1}} &= \norm{-(1-\sqrt{s_1})\ket{\mu_1}\otimes\ket{\nu_1}+\sum_{i\ge 2}\sqrt{s_i}\ket{\mu_i}\otimes\ket{\nu_i}} \\
        &= \sqrt{\br{1-\sqrt{s_1}}^2 + \sum_{i\ge 2}s_i} \\
        &\le \sqrt{\br{1 - \sqrt{1-\varepsilon}}^2 + \varepsilon} \\
        &\le \sqrt{2\varepsilon}.
    \end{align*}
    So
    \begin{equation*}
        \norm{\ketbra{\psi} - \ketbra{\mu_1}\otimes\ketbra{\nu_1}} \le \sqrt{8\varepsilon}
    \end{equation*}
    and

    We let $\ket{\nu} = \ket{\nu_1}$, then
    \begin{align*}
        \norm{\varphi\otimes\nu - \psi} \le \norm{\varphi\otimes\nu - \mu_1\otimes\nu} + \norm{\mu_1\otimes\nu - \psi} \le 2\varepsilon + \sqrt{8\varepsilon}\le 5\sqrt{\varepsilon}.
    \end{align*}
\end{proof}

\end{document}